\documentclass[9pt]{article}
\usepackage[left=0.9in,top=1.2in,right=0.9in,bottom=1.2in,head=1.25in]{geometry}
\usepackage{amsfonts,amsmath,amssymb,amsthm}
\usepackage{verbatim,float,url,enumerate}
\usepackage{graphicx,subfigure,epsfig,psfrag}
\usepackage{dsfont,bm,color,appendix}
\usepackage{natbib}
\usepackage{pdfpages}
\usepackage{algorithm}
\usepackage{algpseudocode}
\usepackage{hyperref}
\usepackage{bbm}
\usepackage[algo2e]{algorithm2e}
\usepackage{enumitem}
\usepackage[thinc]{esdiff}
\usepackage{afterpage}
\usepackage{amsmath,amssymb,bm}
\usepackage{booktabs}
\usepackage{multirow}
\usepackage{array}
\usepackage{xcolor}
\usepackage[symbol]{footmisc}
\usepackage[font=footnotesize,labelfont=bf]{caption}
\usepackage{xcolor}
\usepackage{dblfloatfix}

\newcommand{\mR}{\mathcal{R}}

\newcommand{\norm}[1]{\left\|{#1}\right\|} 
\newcommand{\be}{\begin{equation}}
\newcommand{\ee}{\end{equation}}
\renewcommand{\vec}[1]{\bm{#1}}
\newcommand{\tx}{\vec{\theta_x}}
\newcommand{\tz}{\vec{\theta_z}}
\newcommand{\ttheta}{\vec{\theta}}
\newcommand{\bbeta}{\vec{\beta}}
\newcommand{\y}{\vec{y}}
\newcommand{\z}{\vec{z}}
\newcommand{\1}{\vec{1}}
\newcommand{\thx}{\vec{\hat\theta_x}}
\newcommand{\thz}{\vec{\hat\theta_z}}
\newcolumntype{M}[1]{>{\centering\arraybackslash}m{#1}}

\newcommand{\p}[1]{\left(#1\right)}
\newcommand{\sqb}[1]{\left[#1\right]}

\newcommand{\EE}[2][]{{\rm E}_{#1}\left[#2\right]}
\newcommand{\PP}[2][]{{\rm P}_{#1}\left[#2\right]}
\newcommand{\Var}[2][]{\operatorname{Var}_{#1}\left[#2\right]}
\newcommand{\Cov}[2][]{\operatorname{Cov}_{#1}\left[#2\right]}

\newcommand{\indp}{\;\rotatebox[origin=c]{90}{$\models$}\;}

\newcommand{\gy}{\gamma_y}
\newcommand{\gx}{\gamma_x}
\newcommand{\gz}{\gamma_z}

\newcommand{\thetas}{\theta^{\star}}
\newcommand{\new}{_{\operatorname{new}}}
\def\trans{^{\scriptscriptstyle \sf T}}
\newcommand{\thetah}{\hat{\theta}}

\newcommand{\bx}{\boldsymbol{x}}
\newcommand{\by}{\boldsymbol{y}}
\newcommand{\bz}{\boldsymbol{z}}
\newcommand{\bu}{\boldsymbol{u}}

\newcommand{\oo}{\mathcal{O}}

\theoremstyle{plain}
\newtheorem{prop}{Proposition}

\newtheorem{lemm}{Lemma}

\theoremstyle{definition}

\theoremstyle{remark}

\title{Cooperative Learning for Multiview Analysis}
\author{Daisy Yi Ding$^{2}$, Shuangning Li$^{1}$, Balasubramanian Narasimhan$^{1,2}$, and Robert Tibshirani$^{1,2}$  \\
$^{1}$Department of Statistics, Stanford University  \\ 
$^{2}$Department of Biomedical Data Science, Stanford University}
\date{}

\begin{document}

\maketitle
\begin{abstract}
    We propose a new method for supervised learning with multiple sets of features (``views''). The multiview problem is especially important in biology and medicine, where ``-omics'' data such as genomics, proteomics and radiomics are measured on a common set of samples.
 {\em Cooperative learning} combines the usual squared error loss of predictions with an  ``agreement'' penalty to encourage the predictions from different data views to agree. By varying the weight of the agreement penalty, we get a continuum of solutions that include the
   well-known {\em early} and {\em late fusion} approaches.
   Cooperative learning chooses the degree of agreement (or fusion) in an adaptive manner, using a validation set or cross-validation to estimate test set prediction error.
   One version of our fitting procedure is modular, where one can choose different fitting mechanisms (e.g. lasso, random forests, boosting, neural networks) appropriate for different data views. 
    In the setting of cooperative regularized linear regression, the method combines the lasso penalty with the agreement penalty, yielding feature sparsity.
    The method can be especially powerful when the different data views share some underlying relationship in their signals that can be exploited to boost the signals.
    We show that cooperative learning achieves higher predictive accuracy on simulated data and a real multiomics example of labor onset prediction.
    Leveraging aligned signals and allowing flexible fitting mechanisms for different modalities, cooperative learning offers a powerful approach to multiomics data fusion.
\end{abstract}

\section{Introduction}

With new technologies in biomedicine, we are able to generate and collect data of various modalities, including genomics, epigenomics, transcriptomics, proteomics, and metabolomics (Fig. 1{\em A}).
Integrating heterogeneous features on a common set of observations provides a unique opportunity to gain a comprehensive understanding of an outcome of interest.
It offers the potential for making discoveries that are hidden in data analyses of a single modality and achieving more accurate predictions of the outcome \citep{kristensen2014principles, ritchie2015methods, robinson2017integrative, karczewski2018integrative, ma2020integrative, hao2021integrated}. 
While ``multiview data analysis'' can mean different things, we use it here in the context of supervised learning, where the goal is to fuse different data views to model an outcome of interest.

To give a concrete example, assume that a researcher wants to predict cancer outcomes from RNA expression and DNA methylation measurements for a set of patients.
The researcher suspects that:  (1) both data views potentially have prognostic value; (2) the two views share some underlying relationship with each other, as DNA methylation regulates gene expression and can repress the expression of tumor suppressor genes or promote the expression of oncogenes.
Should the researcher use both data views for downstream prediction, or just use one view or the other?
If using both views, how can the researcher leverage their underlying relationship in making more accurate prediction?
Is there a way to strengthen the shared signals in the two data views while reducing idiosyncratic noise?

There are two broad categories of existing ``data fusion methods'' for the multiview problem (Fig. 1{\em B}).
They differ in the stage at which the ``fusion'' of predictors takes place, namely early fusion and late fusion.
Early fusion works by transforming the multiple data views into a single representation before feeding the aggregated representation into a supervised learning model of choice \citep{yuan2014assessing,gentles2015integrating,perkins2018precision,chaudhary2018deep}. 
The simplest approach is to column-wise concatenate the $M$ datasets $X_{1}, \ldots, X_M$ to obtain a combined matrix $X$, which is then used as the input to a supervised learning model.
Another type of early fusion approach projects each high-dimensional dataset into a low-dimensional space using methods such as principal component analysis or autoencoders \citep{wold1987principal, vincent2010stacked}. Then one combines the low-dimensional representations through aggregation and feed the aggregated matrix into a supervised learning model.
Early fusion approaches have an important limitation that they do not explicitly leverage the underlying relationship across data views.
Late fusion, or ``integration'', refers to methods where individual models are first built from the distinct data views, and then the predictions of the individual models are combined into the final predictor \citep{yang2010review, zhao2019learning, chen2020pathomic, chabon2020integrating, wu2020integrative}. 

\begin{figure}
    \centering
    \includegraphics[width=1\textwidth]{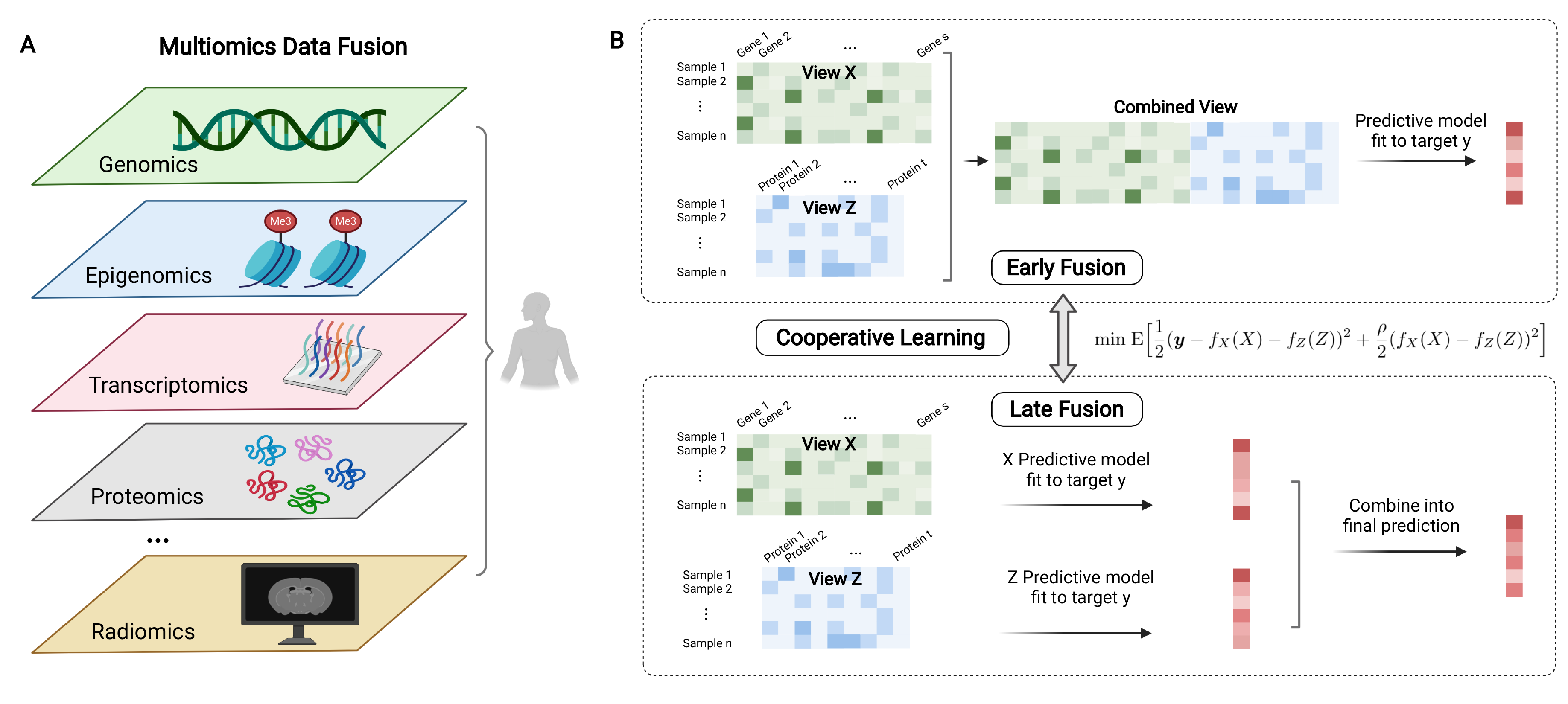}
    \caption{{\em Framework for multiomics data fusion.} {\em (A)} Advances in biotechnologies have enabled the collection of a myriad of ``-omics'' data ranging from genomics to proteomics measured on a common set of samples.
    These data capture the molecular variations of human health at multiple levels and can help us understand complex biological systems in a more comprehensive way.
    Fusing the data offers the potential to improve predictive accuracy of disease phenotypes and treatment response, thus enabling better diagnostics and therapeutics.
    However, multiview analysis of omics data presents challenges such as increased dimensionality, noise and complexity.
    {\em (B)} Commonly-used approaches to the problem can be broadly categorized into early and late fusion. Early fusion begins by transforming all datasets into a single representation, which is then used as the input to a supervised learning model of choice. Late fusion works by developing first-level models from individual data views and then combining the predictions by training a second-level model as the final predictor. Encompassing early and late fusion, cooperative learning combines the usual squared error loss of predictions with an agreement penalty term to encourage the predictions from different data views to align.}
    \label{fig:early_late_fusion}
\end{figure}

In this paper, we propose a new method to multiview data analysis called {\em cooperative learning}, a supervised learning approach that fuses the different views in a systematic way.
The method combines the usual squared error loss of predictions with an ``agreement'' penalty that encourages the predictions from different data views to align.
By varying the weight of the
agreement penalty, we get a continuum of solutions that include the commonly-used early and late
fusion approaches. 
Our proposal can be especially powerful when the different data views share some underlying relationship in their signals that can be leveraged to strengthen the signals.

The rest of the paper is organized as follows.
In Section \ref{sec:coop_learning_two}, we introduce cooperative learning and characterize its solution.
This involves the iterative algorithm for general form of cooperative learning and the explicit closed-form solution for cooperative regularized linear regression.
We discuss its relation with early and late fusion, as well as other existing approaches, and establish theoretical underpinnings of our approach.
In Section \ref{sec:morethan2}, we extend cooperative learning to settings when we have more than two data views.
We demonstrate in Section \ref{sec:simulation_all} the effectiveness of cooperative learning in simulation studies, where we compare it to several commonly-used approaches.
In Section \ref{sec:real}, we apply cooperative learning on a real multiomics example and show that it achieves higher predictive accuracy on labor onset prediction.
In Section \ref{sec:glm_cox}, we discuss how cooperative learning can be extended to generalized linear models and Cox proportional hazards models.
We outline how the framework can be extended to paired features and interaction models in Section \ref{sec:extensions}. 
The paper ends with a discussion and an appendix.

\section{Cooperative learning}
\label{sec:coop_learning_two}
\subsection{Cooperative learning with two data views}
\label{sec:coop_learning_general}
We begin with a simple form of our proposal for the population (random variable) setting.
Let $X \in
\mR^{n \times p_x}$, $Z \in
\mR^{n \times p_z}$ --- representing two data views --- and  $\vec{y} \in \mR^{n}$ be a real-valued random variable (the target). Fixing the
hyperparameter $\rho\geq 0$, we propose to minimize the population quantity:
\begin{equation}
{\rm min} \;  {\rm E}\Bigl[\frac{1}{2}  (\y-f_X(X)-f_Z(Z))^2+ \frac{\rho}{2}(f_X(X)-f_Z(Z))^2\Bigr].
\label{eq:zero0}
\end{equation}

The first term above is the usual prediction error, while
the second term is an ``agreement'' penalty,
encouraging the predictions from different views to agree.
This penalty term is related to ``contrastive learning'' \citep{chen2020simple, khosla2020supervised}, which we discuss in more detail in Section \ref{sec:related_work}.

The solution to \eqref{eq:zero0} has fixed points:

\begin{eqnarray}
f_X(X)&=&{\rm E} \Bigl[\frac{\y}{1+\rho}-\frac{(1-\rho)f_Z(Z)}{(1+\rho)}|X\Bigr],\cr
f_Z(Z)&=&{\rm E} \Bigr[\frac{\y}{1+\rho}-\frac{(1-\rho)f_X(X)}{(1+\rho)}|Z\Bigr].
\label{eq:gen}
\end{eqnarray}

 We can optimize the objective by repeatedly updating the fit for each data view in turn, holding the other view fixed. When updating a function, this approach allows us to apply the fitting method for that data view to a penalty-adjusted ``partial residual''.
 For more than two views, this generalizes easily, with details given in Section \ref{sec:morethan2}.

 The following relationships to early and late fusion can be seen immediately:
 
 \begin{itemize}
     \item   If $\rho=0$, from \eqref{eq:zero0} we see that cooperative learning chooses a functional form for $f_X$ and $f_Z$ and fits them together. If these functions are additive (for example, linear) then it yields  a simple form of early fusion, where we simply use the combined set of features in a supervised learning procedure.
     \item If $\rho=1$, then from \eqref{eq:gen} we see that the solutions
     are the average of the marginal fits for $X$ and $Z$. This is a simple form of late fusion. 
 \end{itemize}
 We explore the relation of cooperative learning to early/late fusion in more detail in Section \ref{sec:rel}, in the setting of regularized linear regression.
 
Note that this ``one-at-a-time'' fitting procedure is modular, so that we can choose a fitting mechanism appropriate for each data view. Specifically:
 
 \begin{itemize}
     \item For {\em quantitative features} like gene expression, copy number variation, or methylation: 
     regularized regression (lasso, elastic net), a generalized additive model, boosting, random forests, or  neural networks.
\item For {\em images}: a convolutional neural network.
\item For {\em time series data}: an auto-regressive model or a recurrent neural network.

\end{itemize}
We illustrate this on a simulated image and omics example in Section \ref{sec:sim_imaging}.

\subsection{Cooperative regularized linear regression}

We make our proposal more concrete in the setting of cooperative regularized linear regression.
Consider feature matrices $X \in
\mR^{n\times p_x}$, $Z \in
\mR^{n\times p_z}$,  and our target $\y \in \mR^{n}$. 
We assume that the columns of $X$ and $Z$ have been standardized, and $\y$ has mean 0 (hence we can omit the intercept below).
For a fixed
value of the hyperparameter $\rho\geq 0$, we want to find $\tx \in \mR^{p_x}$
and $\vec{\tz} \in \mR^{p_z}$ that minimize:

\begin{equation}
J(\tx,\tz) =  \frac{1}{2}  ||\y-X\tx- Z\tz||^2+  \frac{\rho}{2}||(X\tx- Z\tz)||^2+
\lambda_x P^x(\tx)+ \lambda_z P^z(\tz),
\label{eq:obj1}
\end{equation}
where $\rho$ is the hyperparameter that controls the relative importance of the agreement term $||(X\tx- Z\tz)||^2$ in the objective,
and $P^x$ and $P^z$ are penalty functions. 
Most commonly, we use $\ell_1$ penalties, giving the objective function:

\begin{equation}
J(\tx,\tz) =  \frac{1}{2}  ||\y-X\tx- Z\tz||^2+ \frac{\rho}{2}||(X\tx- Z\tz)||^2+
\lambda_x ||\tx||_1+ \lambda_z ||\tz||_1.
\label{eq:obj2}
\end{equation}
Note that when $\rho = 0$, this reduces to early fusion, where we simply concatenate the columns of $X$ and $Z$ and apply lasso.
Furthermore, in Section \ref{sec:rel}, we show that $\rho=1$
yields a late fusion estimate.

In our experiments, we standardize the features and simply take
$\lambda_x=\lambda_z=\lambda$.
We have found that generally there is often no advantage to allowing different $\lambda$ values for different views. 
However, for completeness, in Appendix Section \ref{sec:adap_coop},
we outline an adaptive strategy for optimizing over $\lambda_x$ and $\lambda_z$. We call this {\em adaptive cooperative learning} in our studies. 


With a common $\lambda$, the objective becomes
\begin{equation} 
J(\tx,\tz) =  \frac{1}{2}  ||\y-X\tx- Z\tz||^2+ \frac{\rho}{2}||(X\tx- Z\tz)||^2+
\lambda( ||\tx||_1+  ||\tz||_1), 
\label{eq:sol_full2}
\end{equation}
and we can compute a regularization path of solutions indexed
by $\lambda$.

Problem (\ref{eq:sol_full2})  is convex, and the solution can be computed  as follows. Letting 
\begin{equation}
\tilde X=
\begin{pmatrix}
 X   & Z\\
-\sqrt{\rho}X &  \sqrt{\rho}Z
\end{pmatrix},  
\tilde \y=
\begin{pmatrix}
\y \\
\vec{0}
\end{pmatrix}, 
\tilde \bbeta=
\begin{pmatrix}
\tx \\ \tz  
\end{pmatrix},
\label{eq:sol2}
\end{equation}
then the equivalent problem to \eqref{eq:sol_full2} is 
\begin{equation} 
J(\tx, \tz)=\frac12||\tilde \y-\tilde X\tilde \bbeta||^2+\lambda( ||\tx||_1 +||\tz||_1) .
\label{eq:sol_full}
\end{equation}
This is a form of the lasso, and can be computed, for example
by the {\tt glmnet} package \citep{FHT2010}. 
This new problem has $2n$ observations and $p_x+p_z$ features.

Let ${\rm Lasso}(X,{\vec{y}},\lambda)$ 
 denote the generic problem:
\be
{\rm min}_{\bbeta} \;  \frac12\|{\vec{y}}- X\bbeta\|^2 +
\lambda \|\bbeta\|_1.
\ee
We outline the direct algorithm for cooperative regularized regression in Algorithm \ref{alg:full_direct}.

\begin{algorithm}[t]
\caption{\em Direct algorithm for cooperative regularized regression.}\label{alg:direct_alg}
\KwIn{$X \in \mR^{n\times p_x}$ and $Z \in
\mR^{n\times p_z}$, the response $\vec{y} \in \mR^{n}$, and a grid of hyperparameter values ($\rho_{\tt min}, \ldots, \rho_{\tt max}).$}
\vspace{2mm}

\For{$\rho \gets \rho_{\tt min},\hdots, \rho_{\tt max}$}{
     Set \begin{align*} \tilde X=
\begin{pmatrix}
 X   & Z\\
-\sqrt{\rho}X &  \sqrt{\rho}Z
\end{pmatrix}, \tilde \y=
\begin{pmatrix}
\y \\
\vec{0}
\end{pmatrix}.
\end{align*}
Solve ${\rm Lasso}(\tilde X, \tilde \y, \lambda$) over a decreasing grid of $\lambda$ values. 
}
\vspace{2mm}

Select the optimal value of $\rho^{*}$ based on the CV error and get the final fit.

\label{alg:full_direct}
\end{algorithm}

{\bf Remark A.} We note that for cross-validation (CV) to estimate $\lambda$ and $\rho$, we do not form folds from the rows of $\tilde X$, but instead form folds from the rows of $X$ and $Z$  and then construct the corresponding $\tilde X$.

{\bf Remark B.} We can add $\ell_2$ penalties to
the objective in \eqref{eq:sol_full2}, replacing
$\lambda( ||\tx||_1+||\tz||_1)$ by the elastic net form
\begin{equation}\lambda \Bigl[(1-\alpha)(||\tx||_1+||\tz||_1)+ \alpha(||\tx||_2^2/2+||\tz||_2^2/2)\Bigr].
\end{equation}
This leads to elastic net fitting, in place of the lasso,
in the last step of the algorithm. This option is included in our publicly available software implementation
of cooperative learning.

We show here an illustrative simulation study of cooperative learning in the regression setting in Figure 2{\em A}. 
We will discuss more comprehensive studies in Section \ref{sec:simulation_all}. 
In Figure 2{\em A}, the first and second plots correspond to the settings where the two data views $X$ and $Z$ are correlated, while in the third plot $X$ and $Z$ are uncorrelated.
We see that when the data views are correlated, cooperative learning offers significant performance gains over the early and late fusion methods, by encouraging the predictions from different views to agree.
When the data views are uncorrelated and only one view $X$ contains signal as in the third plot, early and late fusion methods hurt performance as compared to the separate model fit on only $X$, while adaptive cooperative learning is able to perform on par with the separate model.

\begin{figure}[h]
    \centering
    \includegraphics[scale=0.36, width=1.01\textwidth, clip]{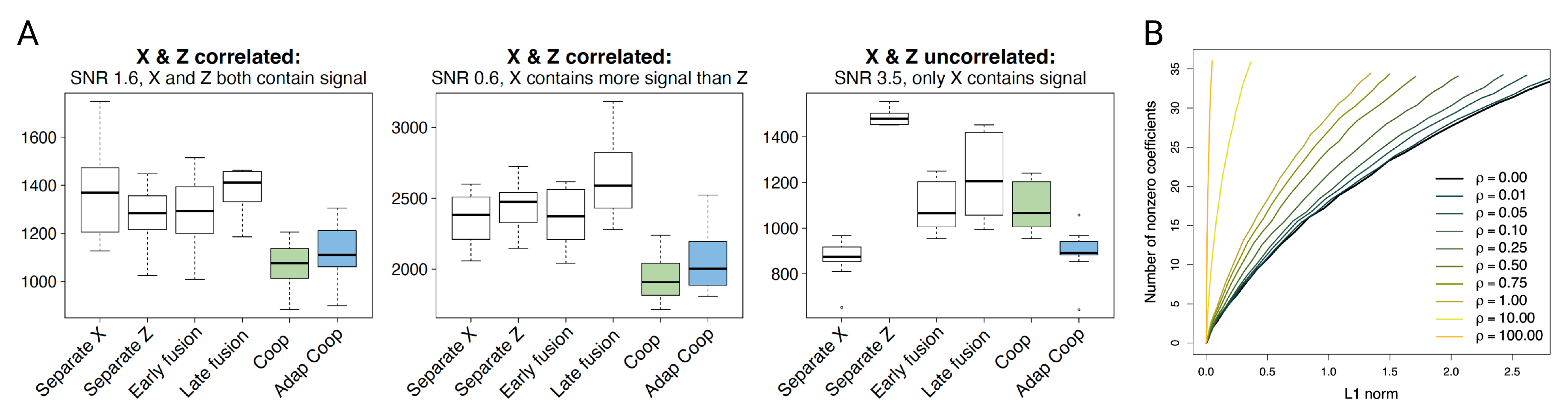}
    \caption{{\em An illustrative simulation study of cooperative learning in the regression setting, and sparsity of the solution.} {\em (A)} Cooperative learning achieves superior prediction accuracy on a test set when the data views $X$ and $Z$ are correlated.
    The y-axis shows the mean squared error (MSE) on a test set.
    The methods in comparison from left to right in each panel correspond to (1) Separate $X$: lasso applied on the data view $X$ only; (2) Separate $Z$: lasso applied on the data view $Z$ only; (3) Early fusion: lasso applied on the concatenated data views of $X$ and $Z$; (4) Late fusion: separate lasso models are fit on $X$ and $Z$ independently and the predictors are then combined through linear least squares; (5) Coop: cooperative learning as outlined in Algorithm \ref{alg:direct_alg}; (6) Adap Coop: adaptive cooperative learning as outlined in Algorithm 
    \ref{alg:adaptive_full} (see Appendix Section \ref{sec:adap_coop}).
    Note that the test MSE in each panel is of a different scale because we experiment with simulating the data of different signal-to-noise ratios (SNR).
    We conducted each simulation experiment 10 times.
    {\em(B)} The number of non-zero coefficients as a function of the $\ell_1$ norm of the solution with different values of the weight on the agreement penalty term $\rho$: the solution becomes less sparse as $\rho$ increases. 
    } 
    \label{fig:illustrative}
\end{figure}

\subsection{ One-at-a-time algorithm for cooperative regularized linear regression}
\label{sec:iterative}
As an alternative, one can optimize \eqref{eq:obj2} by iteratively optimizing over $\tx$ and $\tz$, fixing one and optimizing over the other.  
The updates are as follows:

\begin{eqnarray}
\hat\tx &=& \text{Lasso}(X,\vec{y_x^*},\lambda_x) \text{ , where } \vec{y_x^*}=
\frac{\y}{1+\rho}-\frac{(1-\rho)Z\tz}{(1+\rho)},\cr
\hat\tz &=& \text{Lasso}(Z,\vec{y_z^*},\lambda_z) \text{ , where } \vec{y_z^*}
= \frac{\y}{1+\rho}-\frac{(1-\rho)X\tx}{(1+\rho)}.
\label{eq:oaat}
\end{eqnarray}

This is analogous to the general iterative procedure in (\ref{eq:gen}). It is summarized in Algorithm \ref{alg:one-at-a-time}.

\noindent
\begin{minipage}{.939\textwidth}
\begin{algorithm}[H]
\KwIn{$X \in \mR^{n\times p_x}$ and $Z \in \mR^{n\times p_z}$, the response $\vec{y} \in \mR^{n}$, 
and a grid of hyperparameter values ($\rho_{\tt min}, \ldots, \rho_{\tt max}).$}
\vspace{2mm}
Fix the lasso penalty weights $\lambda_x$ and $\lambda_z$,
\For{$\rho \gets \rho_{\tt min},\hdots, \rho_{\tt max}$}{
Initialize $\tx^{(0)} \in \mR^{p_x}$ and $\tz^{(0)} \in \mR^{p_z}$.

\For{$k \gets 0,1,2,\hdots$ {\normalfont until convergence}}{
\begin{enumerate}
    \item Set $\vec{y_x^*}=\frac{\y}{1+\rho}-\frac{(1-\rho)Z\tz}{(1+\rho)}$. Solve ${\rm Lasso}(X, \vec{y_x^*},\lambda_{x})$
and update $\ttheta_x^{(k+1)}$ to be the solution.
    \item Set $\vec{y_z^*}=\frac{\y}{1+\rho}-\frac{(1-\rho)X\tx}{(1+\rho)}$. Solve ${\rm Lasso}(Z, \vec{y_z^*},\lambda_{z})$ and update $\ttheta_z^{(k+1)}$ to be the solution.
\end{enumerate}
}
}
Select the optimal value of $\rho^{*}$ based on the sum of the CV errors and get the final fit.
\caption{\em One-at-a-time algorithm for cooperative regularized regression.}\label{alg:one-at-a-time}
\end{algorithm}
\end{minipage}

By iterating back and forth between the two lasso problems, we can find the optimal solution to \eqref{eq:obj2}. When both $X$ and $Z$ have full column rank, the problem \eqref{eq:obj2} is strictly convex and each iteration
decreases the overall objective value. Therefore, the one-at-a-time procedure is guaranteed to converge. 
In general, it can be shown to converge to some stationary point, using
results such as those in \cite{tibshirani2017dykstra}.
This algorithm uses  fixed values for  $\lambda_x, \lambda_z$: we need to run the algorithm over  a grid of such values, or use CV to choose $\lambda_x, \lambda_z$
within each iteration.

With just two views, there seems to be no advantage to this approach over the direct solution given in Algorithm \ref{alg:direct_alg}.
However, for a larger number of views, there can be a computational
advantage, which we will discuss in Section \ref{sec:morethan2}.

\subsection{Relation to early/late fusion}
\label{sec:rel}

From the objective functions \eqref{eq:obj1} and  \eqref{eq:obj2},
when the weight on the agreement term $\rho$ is set to 0, cooperative learning (regression) reduces to a form of early fusion:  we simply concatenate the columns of different views and apply lasso or another regularized regression method.

Next we discuss the relation of cooperative learning to late fusion.
Let $X$ and $Z$ have centered columns and $y$ centered, from \eqref{eq:sol2} we obtain 
\begin{equation}
{\tilde X}^T {\tilde X}=
\begin{pmatrix}
X^TX(\1+\rho) & X^TZ (\1-\rho)\\
Z^TX(\1-\rho) &Z^TZ(\1+\rho)\\
\end{pmatrix}.
\label{one}
\end{equation}
Assuming $X$ and $Z$ have full rank, and omitting the $\ell_1$ penalties, we obtain the least squares estimates
\begin{equation}
\begin{pmatrix}
\thx\cr
\thz\cr
\end{pmatrix}
=
\begin{pmatrix}
X^TX(1+\rho) & X^TZ (1-\rho)\\
Z^TX(1-\rho) &Z^TZ(1+\rho)\\
\end{pmatrix}
 ^{-1} 
\begin{pmatrix}
X^T \y\cr
Z^T \y\cr
\end{pmatrix}.
\end{equation}
If  $X^TZ=0$ (uncorrelated features between the views), this reduces to a linear combination of the least squares estimates for each block; when $\rho=1$, it is simply the average of  the least squares estimates for each block.
The above relation also holds when we include the $\ell_1$ penalties.

This calculation  suggests that  restricting  $\rho$ to be in $[0,1]$ would be natural. However, we have found that values larger than one can sometimes yield lower prediction error (see the simulation studies in Section \ref{sec:simulation_all}).

\subsection{Sparsity of the solution}
\label{sec:sparse}


We explore how the sparsity of the solution depends on the agreement hyperparameter $\rho$ in Figure 2{\em B}.
We did 100 simulations of Gaussian data with $n=100$ and $p=20$ in each of two views, with all coefficients equal to 2.0. The standard deviation (SD) of the errors was chosen so that the SNR was about 2.
The figure shows the number of non-zero coefficients as a function of the overall $\ell_1$ of the solutions, for different values of $\rho$. Note that the lasso parameter $\lambda$ is varying along the horizontal axis; we chose to plot against the $\ell_1$ norm, a more meaningful quantity.
We see that the solutions become less sparse as $\rho$ increases, much like the behavior that one sees in the elastic net.

\subsection{Theoretical analysis under the latent factor model}
To understand the role of the agreement penalty from a theoretical perspective, we consider the following latent factor model. Let $\bu = (U_1, U_2, \dots, U_n)$ be a vector of $n$ i.i.d. random variables with $U_i \sim \mathcal{N}(0,1)$, $\by = (y_1, \dots, y_n)$, $\bx = (X_1, \dots, X_n)$, and $\bz = (Z_1, \dots, Z_n)$, with $y_i = \gy U_i + \varepsilon_{yi}$, $X_i = \gx U_i + \varepsilon_{xi} \text{ and } Z_i = \gz U_i + \varepsilon_{zi} $,
where $\varepsilon_{yi} \sim \mathcal{N}\p{0,\sigma_y^2}$, $\varepsilon_{xi} \sim \mathcal{N}\p{0,\sigma_x^2}$, $\varepsilon_{zi} \sim \mathcal{N}\p{0,\sigma_z^2}$ independently. 
We show that the mean squared error (MSE) of the predictions from cooperative learning is a decreasing function of $\rho$ around 0 with high probability (see details in Appendix Section \ref{sec:appendix_theoretical_analysis}).
Therefore, the agreement penalty offers an advantage in reducing MSE of the predictions under the latent factor model.

\subsection{Relation to existing approaches}
\label{sec:related_work}
We have mentioned the close connection of cooperative learning to {\em early and late fusion}: setting $\rho=0$ or 1 gives a version of each of these, respectively.
There are many variations of late fusion, including
the use of stacked generalization to combine the predictions at the last stage \citep{GARCIACEJA201845}.

Cooperative learning is also related to {\em collaborative regression} \citep{gross2015collaborative}.
This method uses an objective function of the form
\begin{equation}
    \frac{b_{xy}}{2}||\y-X\tx||^2+ \frac{b_{zy}}{2}||\y-Z\tz||^2+ \frac{b_{xz}}{2}||X\tx-Z\tz||^2.
\end{equation}
With $\ell_1$ penalties added, this is proposed as a method for sparse supervised canonical correlation analysis.
It is different from cooperative learning in an important way:
here $X$ and $Z$ are not fit jointly to the target. The authors state that 
collaborative regression is not well suited
to the prediction task.
We note that if $b_{xy}=b_{zy}=b_{xz}=1$, each of $\thx, \thz$ are
the one-half of the least squares (LS) estimates on $X, Z$ respectively. Hence the overall prediction $\hat\y$ is the average of the individual LS predictions.
This late fusion estimate is the same as that obtained from cooperative learning with $\rho=1$.
In addition, a related framework based on optimizing measures of agreement between data views was also proposed in \cite{sindhwani2005co}, but it is different from cooperative learning in the sense that the data views are not used jointly to model the target.

Cooperative learning also has connections with {\em contrastive learning} \citep{chen2020simple, khosla2020supervised}. 
This method is an unsupervised learning technique first proposed for learning visual representations.
Without the supervision of $\y$, it learns representations of images by maximizing agreement between differently augmented ``views'' of the same data example.
While both contrastive learning and cooperative learning have a term in the objective that encourages agreement between correlated views, our method combines the agreement term with the usual prediction error loss and is thus supervised.

Moreover, the iteration 
\eqref{eq:gen}
looks much like the {\em backfitting} algorithm
for generalized additive models \citep{hastie1990generalized}. 
In that setting,
 each of $f_X$ and $f_Z$ are typically functions
 of one-dimensional features $X$ and $Z$, and the backfitting algorithm iterations correspond to
\eqref{eq:gen}
with $\rho=0$.
In the additive model setting, backfitting is a special case of the Gauss-Seidel algorithm \citep{hastie1990generalized}.
In cooperative learning, each of $X,Z$ are views with multiple features; we could use an additive model for each view, i.e. $f_X(X)=\sum_i g_i(X_i)$, $f_Z(Z)=\sum_j h_j(Z_j)$, where $i$ and $j$ are column indices of $X$ and $Z$, respectively. Then each of the iterations in 
\eqref{eq:gen}
could be solved using a backfitting algorithm, leading to a nested procedure.

We next discuss the relation of cooperative learning to a recently proposed method for multiview analysis called {\em sparse integrative discriminant analysis (SIDA)} \citep{safo2021sparse}.
This method aims to identify variables that are associated across views while also able  to optimally separate data points into different classes.
Specifically, it combines canonical correlation analysis and linear discriminate analysis by solving the following optimization problem. 
Let $X_{k} = (\vec{x}_{1k},\ldots, \vec{x}_{n_{k},k})^{T} \in \mR^{n_{k} \times p}$, $\vec{x}_{k} \in \mR^{p}$ be the data matrix for class $k$, where $k=1,\ldots,K$, and $n_{k}$ is the number of samples in class $k$.
Then, the mean vector for class k is $\hat{\mu}_{k} = \frac{1}{n_{k}}\sum_{i=1}^{n_{k}} \vec{x}_{ik}$; the common variance matrix for all classes is $S_{w} = \sum_{k=1}^{K}\sum_{i=1}^{n}(\vec{x}_{ik}-\hat{\mu}_k)(\vec{x}_{ik}-\hat{\mu}_k)^{T}$; the between class covariance matrix is $S_{b} = \sum_{k=1}^{K}n_{k}(\hat{\mu}_k - \hat{\mu})(\hat{\mu}_k - \hat{\mu})^{T}$, where $\hat{\mu} = \frac{1}{n}\sum_{k=1}^{K}n_{k}\hat{\mu}_k$ is the combined class mean vector.
Assume that we have two data views $X \in \mR^{n \times p_{x}}$ and $Z \in \mR^{n \times p_{z}}$ with centered columns, we want to find $A = [\vec{a}_{1},\ldots,\vec{a}_{K-1}]$ and $B = [\vec{b}_{1},\ldots,\vec{b}_{K-1}]$ such that
\vspace{-1mm}
\begin{align*}
    &{\rm max} \; \rho \cdot \text{tr}(A^{T}S_{b}^{x}A + B^{T}S_{b}^{z}B) + (1 - \rho)\cdot\text{tr}(A^{T}S_{xz}BB^{T}S^{T}_{xz}A) \nonumber \\
    &\text{s.t. } \text{tr}(A^{T}S_{w}^{x}A) / (K-1) = 1 \text{ \& } \text{tr}(B^{T}S_{w}^{z}B) / (K-1) = 1,
\end{align*}
where $S_{xz} \in \mR^{p_{x} \times p_{z}}$ is the sample cross-covariance matrix between $X$ and $Z$. Here, tr($\cdot$) is the trace function, and $\rho$ is the parameter that controls the relative importance of the ``separation'' term and the ``association'' terms in the objective.
While SIDA also considers the association across data views by choosing vectors that are associated and able to separate data points into classes, it solves the problem in a ``backward'' manner, that is the features are modeled as a function of the outcome.
Cooperative learning, in contrast, solves the problem in a ``forward'' manner ($Y \sim X,Z)$, which is more suitable for
prediction.

We also note the connection between cooperative learning (regression) with the  {\em standardized group lasso} \citep{simon2012standardization}.
This method is a variation of the group lasso \citep{yuan2006model}, and uses 
\begin{equation}
    \|X\tx\|_{2} + \|Z\tz\|_{2}\label{eq:sgl}
\end{equation} 
as the penalty term, rather than the sum of squared two norms. 
It encourages group-level sparsity by eliminating entire blocks of features at a time.
In the group lasso,  each block is a group of features, and we do not expect each block to be predictive on its own. 
This is different from cooperative learning, where each feature block is a data view and we generally do not want to eliminate an entire view for prediction. In addition, the standardized group lasso
does not have an agreement penalty. 
One could in fact add the
standardized group lasso penalty (\ref{eq:sgl}) to the cooperative learning objective, which would allow elimination of entire data views.

\section{Cooperative learning with more than two data views}
\label{sec:morethan2}
When we have more than two views of the data, $X_1 \in
\mR^{n\times p_1},X_2\in
\mR^{n\times p_2},\ldots, X_M\in
\mR^{n\times p_M}$, the population quantity that we want to minimize becomes 
\begin{equation}
{\rm min} \;  
{\rm E}\Bigl[\frac{1}{2}(\y-\sum_{m=1}^{M} f_{X_{m}}(X_{m}))^2 + \frac{\rho}{2}\sum_{m<m'}(f_{X_{m}}(X_{m})-f_{X_{m'}}(X_{m'}))^2\Bigr].
\label{eq:generalK}
\end{equation}
We can also have different weights on the agreement penalties for distinct pairs of data views, forcing some pairs to agree more than others. In addition, we can incorporate prior knowledge in determining the relative strength of the agreement penalty for each pair of data views.


As with two views, this can be optimized with an iterative algorithm that updates each $f_{X_m}(X_m)$ as follows:
\begin{equation}
f_{X_m}(X_m)={\rm E} \Bigl[\frac{\y}{1+(M-1)\rho} -\frac{(1-\rho)\sum_{m'\neq m} f_{X_{m'}(X_{m'})}}{1+(M-1)\rho}|X_m\Bigr].
\label{eq:gen2}
\end{equation}
As in the two-view setup above, the fitter $E(\cdot | X_m)$ can be tailored to the data type of each view.

For regularized linear regression with more than two views,  the objective becomes
\begin{equation}
J(\ttheta_1,\ttheta_2, \ldots, \ttheta_M) = 
\frac{1}{2} ||\y- \sum_{m=1}^{M} X_m\ttheta_m||^2+  \frac{\rho}{2}\sum_{m<m'} ||(X_m\ttheta_m- X_{m'}\ttheta_{m'})||^2 + \sum_{m=1}^{M} \lambda_m \|\ttheta_m||_1.
\label{eq:objK}
\end{equation}

This is again a convex problem.
The optimal solution can be found by forming augmented data matrices as before in \eqref{eq:sol2} and \eqref{eq:sol_full}.

Let \begin{equation*}
\tilde X=
\begin{pmatrix}
 X_{1}   &  X_{2} & ...& X_{M-1} & X_{M}\\
-\sqrt{\rho}X_{1} &  \sqrt{\rho}X_{2} & ... & 0 & 0 \\
-\sqrt{\rho}X_{1} &  0 & ... & \sqrt{\rho}X_{M-1} & 0 \\
-\sqrt{\rho}X_{1} &  0 & ... & 0 & \sqrt{\rho}X_{M} \\
0 &  -\sqrt{\rho}X_{2}  & ... & \sqrt{\rho}X_{M-1}& 0 \\
0 &  -\sqrt{\rho}X_{2} & ... & 0 & \sqrt{\rho}X_{M} \\
... & ... & ...& ...&...\\
0 & 0 & ... & -\sqrt{\rho}X_{M-1} & \sqrt{\rho}X_{M} \\
\end{pmatrix}, 
\end{equation*}

\begin{equation}
\tilde\y=
\begin{pmatrix} 
\y&\vec{0}&...&\vec{0} 
\end{pmatrix}^T, \;
\tilde\bbeta=
\begin{pmatrix}
\ttheta_1&\ttheta_2&...&\ttheta_M  
\end{pmatrix}^T,
\label{eq:solK_full_matrix}
\end{equation}
then the equivalent problem to \eqref{eq:objK} becomes 
\begin{equation} 
\frac12||\tilde \y-\tilde X\tilde \bbeta||^2+
\sum_{m=1}^{M} \lambda_m \|\ttheta_m||_1. 
\label{eq:solK_full}
\end{equation}

With $M$ views, the augmented matrix in \eqref{eq:solK_full_matrix} has $n+{M\choose 2}\cdot n$ rows, which could be computationally challenging to solve.
Alternatively, 
the optimal solution $\hat{\ttheta_1},\hat{\ttheta_2}, \ldots, \hat{\ttheta_M}$  has fixed points
\begin{eqnarray}
\hat\ttheta_m &=& \text{Lasso}(X,\vec{y_m^*},\lambda_{m}) \text{ , where } \vec{y_m^*}
= \frac{\y}{1+(M-1)\rho}-\frac{(1-\rho)\sum_{m'\neq m} X_{m'}\ttheta_{m'}}{1+(M-1)\rho}.
\label{eq:solK}
\end{eqnarray}

This leads to an iterative algorithm, where we successively solve each subproblem, until convergence. For a large number of views, this can be a more efficient procedure than the direct approach in \eqref{eq:solK_full} above.
We include simulation studies on cooperative learning for more than two views in Appendix Section \ref{sec:appendix_more_than_two_views}.

\section{Simulation studies} 
\label{sec:simulation_all}

\subsection{Simulation studies on cooperative regularized linear regression}
\label{sec:simulation}
Here we compare cooperative learning in the regression setting with early and late fusion methods in simulations. 
The set up is as follows.
Given values for parameters $n, p_{x}, p_{z}, p_{u}, s_{u}, t_{x}, t_{z}, \bbeta_{u}, \sigma$, we generate data according to the following procedure: 
\begin{enumerate}[noitemsep]
    \item $x_j \in \mR^{n}$ distributed i.i.d. MVN$(0, I_{n})$ for $j = 1,2,\ldots,p_{x}$.
    \item $z_j \in \mR^{n}$ distributed i.i.d. MVN$(0, I_{n})$ for $j = 1,2,\ldots,p_{z}$.
    \item For $i = 1, 2, \ldots, p_{u}$ ($p_{u}$ corresponds to the number of latent factors, $p_{u} < p_{x}$ and $p_{u} < p_{z}$):
    \begin{enumerate}
        \item $u_{i} \in \mR^{n}$ distributed i.i.d. MVN$(0, s_{u}^2I_{n})$;
        \item $x_{i} = x_{i} + t_{x} * u_{i}$;
        \item $z_{i} = z_{i} + t_{z} * u_{i}$.
    \end{enumerate}
    \item $X=[x_{1}, x_{2}, \ldots, x_{p_{x}}]$, $Z = [z_{1}, z_{2}, \ldots, z_{p_{z}}]$.
    \item $U = [u_{1}, u_{2}, \ldots, u_{p_{u}}]$, $\y = U\bbeta_{u} + \epsilon$ where $\epsilon \in \mR^{n}$ distributed i.i.d. MVN$(0, \sigma^2 I_{n})$.
\end{enumerate}

\begin{figure}[h!]  
\centering
\includegraphics[width=1\textwidth]{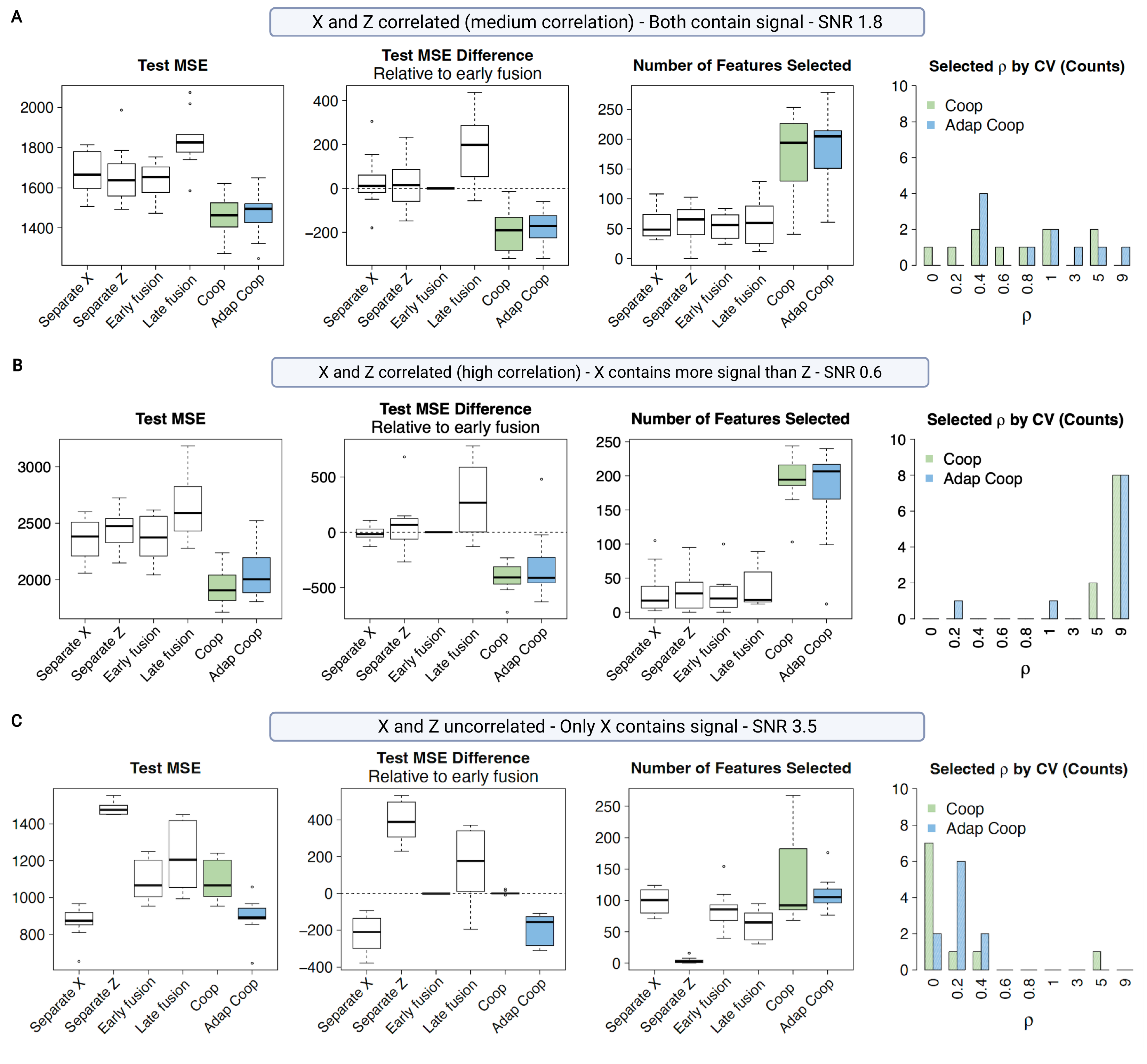}
\caption{{\em Simulation studies on cooperative regularized linear regression.} {\em (A)} Simulation results when $X$ and $Z$ have a medium level of correlation and both contain signal ($t_{x} = t_{z} = 2$), $n = 200, p = 1000$, SNR $= 1.8$. The first panel shows MSE on a test set; the second panel shows the MSE difference on the test set relative to early fusion; the third panel shows the number of features selected; the fourth panel shows the $\rho$ values selected by CV in cooperative 
learning. Here ``Coop'' refers to cooperative 
learning outlined in Algorithm \ref{alg:direct_alg} and ``Adap Coop'' refers to adaptive cooperative learning outlined in Algorithm \ref{alg:adaptive_full} (see Appendix Section \ref{sec:adap_coop}). 
{\em (B)} Simulation results when $X$ and $Z$ have a high level of correlation and X contains more signal than Z ($t_{x} = 6, t_{z} = 1$), $n = 200, p = 1000$, SNR $= 0.6$.
{\em (C)}  Simulation results when $X$ and $Z$ have no correlation; only $X$ contains signal ($t_{x} = 2, t_{z} = 0$), $n = 200, p = 1000$, SNR $= 3.5$.
} \label{fig:res_1}
\end{figure}

There is sparsity in the solution since a subset of columns of $X$ and $Z$ are independent of the latent factors used to generate $\y$.
Data sets are simulated with different levels of correlation between the two data views $X$ and $Z$, different contributions of $X$ and $Z$ to the signal, and different signal-to-noise ratios (SNR).
We consider the settings of both small $p$ and large $p$ regimes, and of both low and high SNR ratios.
We use 10-fold CV to select the optimal values of hyperparameters.
We compare the following methods: 
\begin{itemize}
    \item Separate $X$ and separate $Z$: The standard lasso is applied on the separate data views of $X$ and $Z$ with 10-fold CV. 
    \item Early fusion: The standard lasso is applied on the concatenated data views of $X$ and $Z$ with 10-fold CV. Note that this is equivalent to cooperative learning with $\rho = 0$.
    \item Late fusion: Separate lasso models are first fitted on $X$ and $Z$ independently with 10-fold CV, and the two resulting predictors are then combined through linear least squares for the final prediction.
    \item Cooperative learning (regression) and adaptive cooperative learning.
\end{itemize}

We evaluated the performance based on the mean-squared error (MSE) on a test set.
We conducted each simulation experiment 10 times.

Overall, the simulation results can be summarized as follows:
\begin{itemize}
\item Cooperative learning performs the best in terms of test MSE across the range of SNR and correlation settings.
It is most helpful when the data views are correlated and both contain signal (as in Figure 3{\em A} and Figure 3{\em B}). 
When the correlation between data views is higher, higher values of $\rho$ are more likely to be selected.


\item When only one view contains signal and the views are not correlated (as in Figure 3{\em C}), cooperative learning is outperformed by the separate model fit on the view containing the signal, but adaptive cooperative learning is able to perform on par with the separate model, outperforming early and late fusion.

\item Moreover, we also find that cooperative learning tends to yield a less sparse model, as expected from the results of Section \ref{sec:sparse}.
\end{itemize}
We include more comprehensive results across a wider range of simulation settings in Section \ref{sec:appendix_sim} in the Appendix.

\subsection{Simulation studies on cooperative learning with imaging and ``omics" data}
\label{sec:sim_imaging}
Here we extend the simulation studies for cooperative learning to the setting where we have two data views of more distinct data modalities, such as imaging and omics data (e.g. transcriptomics and proteomics). 
We tailor the fitter suitable to each view, i.e. convolutional neural networks (CNN) for images and lasso for omics.
We simulate the ``omics'' data ($X$) and the ``imaging'' data ($Z$) such that they share some common factors. 
These factors are also used to generate the signal in the response $\y$.
We use a factor model to generate the data, as it is a natural way to create correlations between $X, Z,$ and $\y$.
In Appendix Section \ref{sec:appendix_imaging}, we outline the full details of the simulation procedure.
Figure 4 shows some examples of the synthetic images generated for this study.

\begin{figure}[h]  
\centering
\includegraphics[width=0.5\textwidth]{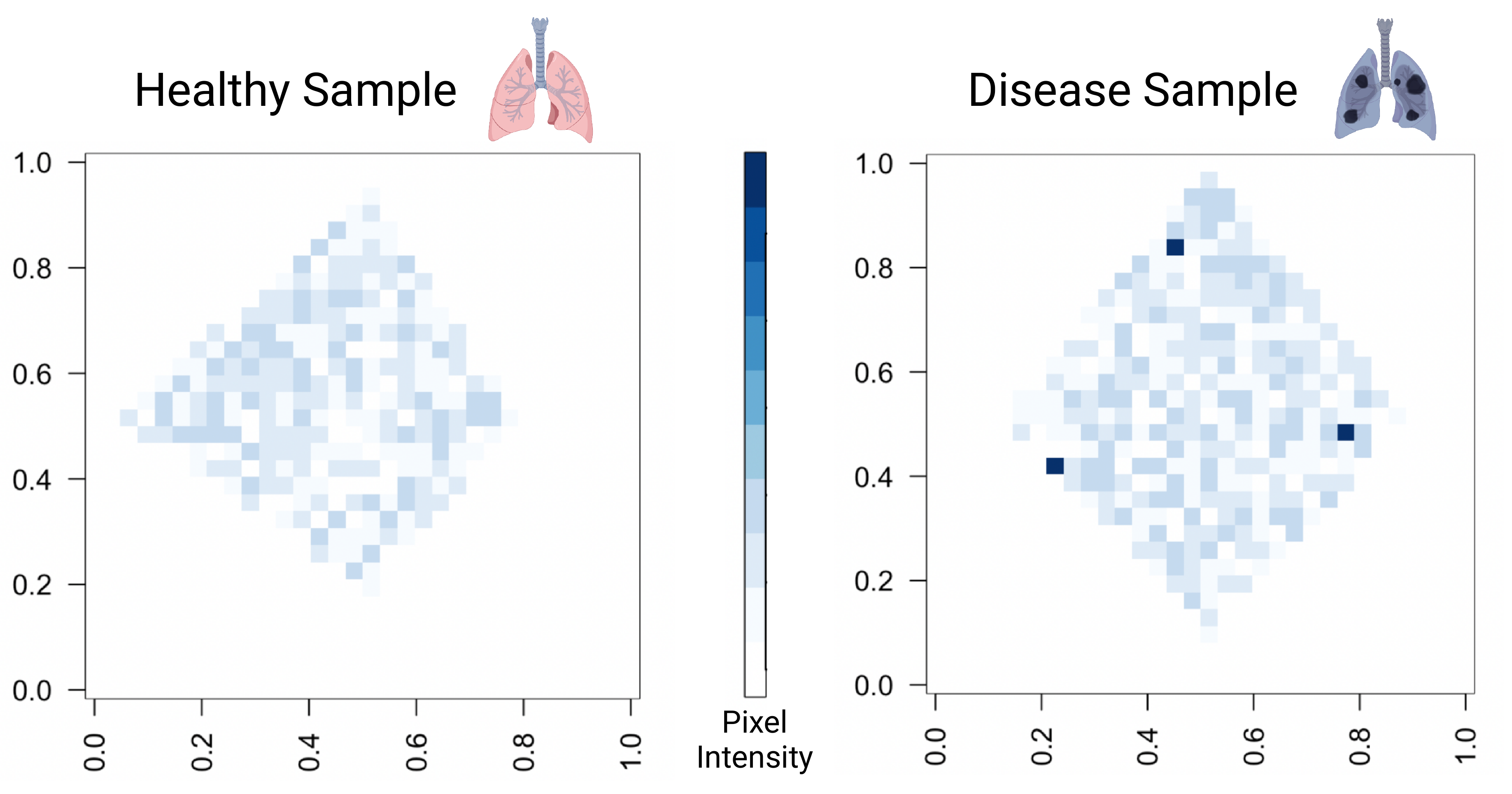}
\caption{{\em Generated images for ``healthy'' and ``disease'' samples.} One can think of the image as an abstract form of a patient's lung, with the darker spots corresponding to the tumor sites. The intensity of the dark spots on the disease samples is generated to correlate with the omics data and the signal in the outcome.}
\label{fig:images}
\end{figure}

Our task is to use the omics and imaging data to predict if a patient has a certain disease.
We use CNN for modeling the imaging data and lasso for the omics data, and optimize the objective for the general form of cooperative learning as in Algorithm \ref{eq:zero0} with the iterative ``one-at-a-time'' algorithm outlined in Algorithm \ref{eq:gen}.

We compare cooperative learning to the following methods: (1) Only images: a simple one-layer CNN with max pooling and rectified linear unit (ReLU) activation is applied on the imaging data only; (2) Only omics: the standard lasso is applied on the omics data only; (3) Late fusion: separate models (CNN and lasso) are first fit on the imaging and omics data, respectively, and the resulting predictors are then combined through linear least squares using a validation set. 
We evaluated the performance based on the misclassification error on a test set, as well as the difference in misclassification error relative to late fusion\footnote{Early fusion is not applicable in this setting.}. 
We consider both low and high SNR settings\footnote{The SNR is calculated based on the logits of the probabilities used to generate the class labels.}.
We conducted each simulation experiment 10 times.

The results are shown in Figure \ref{fig:omics}.
We find that (1) late fusion achieves a lower misclassification error on the test set than the separate models; (2) cooperative learning outperforms late fusion and achieves the lowest test error by encouraging the predictions from the two views to agree; (3) cooperative learning is especially helpful when the SNR is low, while its benefit is less pronounced when the SNR is higher. 
The last observation makes sense, because when the SNR is lower the marginal benefit of leveraging the other view(s) in strengthening signal becomes larger.


\section{Real multiomics studies} 
\label{sec:real}

We applied cooperative learning (regression) to a data set of labor onset, collected from a cohort of women who went into labor spontaneously, as described in \cite{stelzer2021integrated}.
Proteome and metabolome were measured from blood samples collected from the patients during the last 120 days of pregnancy. 
The goal of the analysis is to predict time to spontaneous labor using proteomics and metabolomics data.

\begin{figure*}[h!]
\centering
\includegraphics[width=1\textwidth]{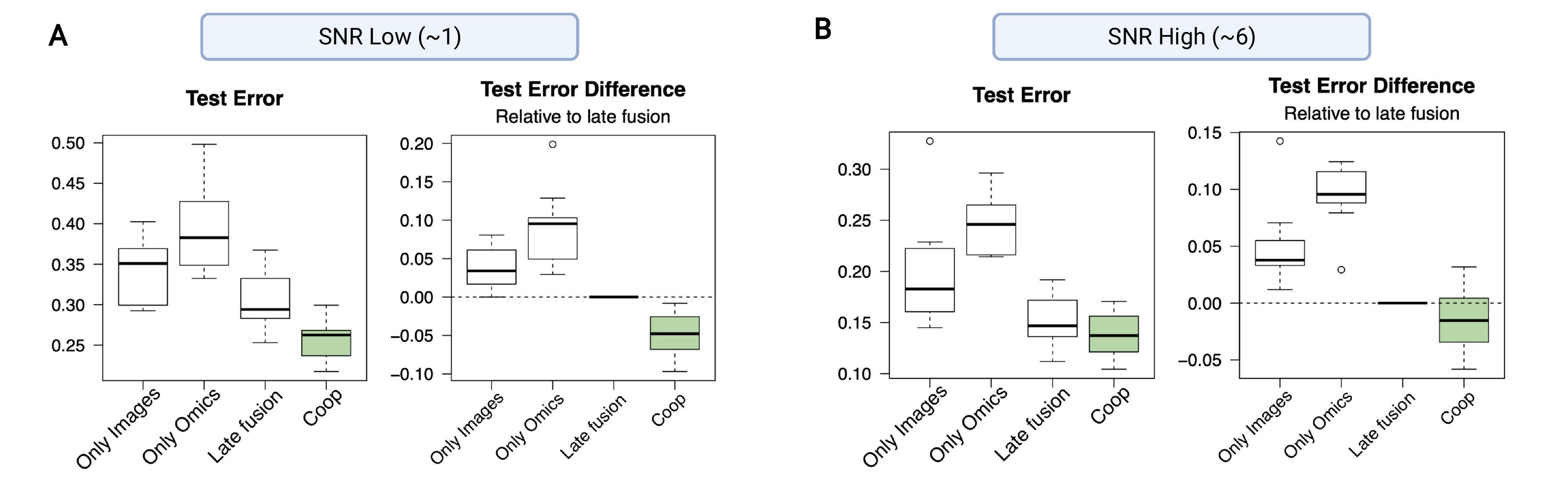}
\caption{{\em Simulation studies on cooperative learning with imaging and ``omics'' data}. Panel {\em (A)} corresponds to the relatively low SNR setting (SNR = 1) and panel {\em (B)} to the higher SNR setting (SNR = 6). For each setting, the left panel shows the misclassification error on the test set for CNN on only images, lasso on only omics, late fusion, and cooperative learning; the right panel shows the difference in misclassification error relative to late fusion. Here ``Coop'' refers to cooperative learning. For both settings, the range of $\rho$ values for cooperative learning to select from is (0,20). The average $\rho$ selected in the low SNR setting is 6.8 and in the high SNR setting is 8.0. 
}
\label{fig:omics}
\end{figure*}

The proteomics data contained measurements for 1,322 proteins and the metabolomics data contained measurements for 3,529 metabolites.
We split the dataset of 53 patients into training and test sets of 40 and 13 patients, respectively\footnote{The cohort consisted of 63 patients as described in \cite{stelzer2021integrated}, but in the public dataset we only found 53 patients with matched proteomics and metabolomics data.}. 
Both the proteomics and metabolomics measurements were screened by their variance across the subjects.
We extracted the first time point for each patient from the longitudinal study and predicted the corresponding time to labor.
We conducted the same set of experiments across 10 different random splits of the training and test sets.

\begin{table*}[!b]
    \small
    \begin{center}
    \vspace{-1mm}
    \begin{tabular}{c|M{1.6cm}M{1.6cm}|M{1.6cm}M{1.6cm}|M{3.2cm}}
        \toprule
        \textbf{Methods} & \multicolumn{2}{c|}{\textbf{Test MSE}} & \multicolumn{2}{c|}{\textbf{Relative to Early Fusion}} & \textbf{Number of Features Selected} \\
        & Mean & SD & Mean & SD & Mean \\ \hline
        \midrule
        Separate Proteomics & 475.51 & 80.89 & 69.14 & 81.44 & 26 \\ 
        Separate Metabolomics & 381.13 & 36.88 & -25.24 & 30.91 & 11 \\
        Early fusion & 406.37 & 44.77 & 0 & 0 & 15 \\
        Late fusion & 493.34 & 63.44 & 86.97 & 68.13 & 21\\
        \textbf{Cooperative learning} & \textbf{335.84} & \textbf{38.51}  &  \textbf{-70.53} & \textbf{32.60} & 52 \\
        \bottomrule
    \end{tabular}
    \end{center}
    \caption{{\em Multiomics studies on labor onset prediction.} The first two columns in the table show the mean and standard deviation (SD) of MSE on the test set across different splits of the training and test sets; the third and fourth column show the MSE difference relative to early fusion; the last column shows the average number of features selected. The methods include (1) separate proteomics: the standard lasso is applied on the proteomics data only; (2) separate metabolomics: the standard lasso is applied on the metabolomics data only; (3) early fusion: the standard lasso is applied on the concatenated data of proteomics and metabolomics data; (4) late fusion: separate lasso models are first fit on proteomics and metabolomics independently and the predictors are then combined through linear least squares; (5) cooperative learning (Algorithm \ref{alg:direct_alg}).  The average of the selected $\rho$ values is 0.9 for cooperative learning.}
    \label{tab:results_labor_onset}
\end{table*}

The results are shown in Table \ref{tab:results_labor_onset}.
The model fit on the metabolomics data achieves lower test MSE than the one fit on the proteomics data.
Early and late fusion hurt performance as compared to the model fit on only metabolomics.
Cooperative learning gives performance gains over the model fit only on metabolomics, outperforming both early and late fusion and achieving the lowest MSE on the test set.

We examined the selected features from cooperative learning and the other methods by comparing the ranking of the features based on the magnitude of their coefficients. 
All methods rank sialic acid binding immunoglobulin like lectin-6 (Siglec-6), a protein highly expressed by the placenta \citep{brinkman2007human}, as the most important feature for predicting labor onset.
As compared to the other methods, cooperative learning boosts up the ranking of features such as plexin-B2 (PLXNB2), which is a protein expressed by the fetal membranes \citep{singh2015endometrial}, and Activin-A, which is highly expressed by the placenta as well \citep{stelzer2021integrated}.
While factors such as Siglec-6, PLXNB2 and Activin-A have previously also been discovered by \cite{stelzer2021integrated} for labor onset prediction, C1q was only identified by cooperative learning as one of the top ten features.
C1q is an important factor involved in the complement cascade, which influences implantation and fetal development \citep{girardi2020essential}, and worth further investigation for its role in predicting labor onset.

\section{Cooperative generalized linear models and Cox regression}
\label{sec:glm_cox}
We next describe how
cooperative learning can be extended to generalized linear models (GLMs) \citep{nelder1972generalized} and Cox proportional hazards models \citep{cox1972regression}.

Consider a GLM, consisting  of 3 components: (1) a linear predictor: $\eta = X\bbeta$; (2) a link function $g$ such that ${\rm E}(Y | X) = g^{-1}(\eta)$; (3) a variance function as a function of the mean: $V = V({\rm E}(Y|X))$. For cooperative GLMs, we have the linear predictor as $\eta = X\tx + Z\tz$, and an additional agreement penalty term $\rho||(X\tx- Z\tz)||^2$ with the following objective to be minimized: 
\vspace{-2mm}

\begin{equation}
    J(\tx, \tz)=
\ell(X\tx + Z\tz, \vec{y})+\frac{\rho}{2}||(X\tx- Z\tz)||^2 +\lambda_{x} ||\tx||_1 + \lambda_{z}||\tz||_1,
\label{eq:glm}
\end{equation} 
where $\ell$ is the negative log likelihood (NLL) of the data.
For Cox proportional hazards models, $\ell$ becomes the negative log partial likelihood of the data. 

We make the usual quadratic approximation to \eqref{eq:glm}, reducing the minimization problem to a weighted least squares (WLS) problem, which yields
\begin{equation}
{\rm min} \;  \frac{1}{2} [ ||W(\z-X\tx- Z\tz)||^2+ \rho||(X\tx- Z\tz)||^2] + \lambda_x ||\tx||_1+ \lambda_z ||\tz||_1,
\label{eq:wls}
\end{equation}
where $\z$ is the adjusted dependent variable and $W$ is the diagonal weight matrix, both of which are functions of $\tx$ and $\tz$.

This leads to an iteratively reweighted least squares (IRLS) algorithm:
\begin{itemize}
    \item Outer loop: Update the quadratic approximation using the current parameter $\hat{\ttheta_x}$ and $\hat{\ttheta_z}$, i.e. update the working response $\vec{z}$ and the weight matrix $W$.
    \item Inner loop: Letting 
\begin{equation}
\tilde X=
\begin{pmatrix}
 W^{1/2}X   & W^{1/2}Z\\
-\sqrt{\rho}X &  \sqrt{\rho}Z
\end{pmatrix},  
\tilde{\vec{z}}=
\begin{pmatrix}
W^{1/2}\vec{z} \\
\vec{0}
\end{pmatrix}, 
\tilde \bbeta=
\begin{pmatrix}
\tx \\ \tz  
\end{pmatrix},
\label{eq:sol_irls}
\end{equation}
solve the following problem
\begin{equation} 
J(\tx, \tz)=\frac12||\tilde{\vec{z}}-\tilde X\tilde \bbeta||^2+\lambda_x ||\tx||_1+ \lambda_z ||\tz||_1,
\label{eq:obj_irls}
\end{equation}
 which is equivalent to \eqref{eq:wls}.
\end{itemize}

\section{Some extensions}
\label{sec:extensions}
\subsection{Paired features from different views}
\label{sec:paired}
One can extend cooperative learning to  the setting where 
a feature in one view is naturally paired with a feature in another view. For example, if the $j$th column $X_j$ of $X$ is the gene expression for gene $j$, and $Z_{k}$ is the
expression of the protein $k$ for which gene $j$ codes. In that setup, we would like to encourage agreement between $X_j\ttheta_{xj}$ and $Z_{k}\ttheta_{zk}$.
This pairing need not exist for all features, but can occur for a subset of features.

Looking back at our objective function \eqref{eq:obj2} for two views in the linear case, we add to this objective a pairwise agreement penalty of the form
\begin{equation}
    \rho_2\sum_{j,k \in P} (X_j\ttheta_{xj}-Z_{k}\ttheta_{zk})^2
\end{equation}
where $P$ is the set of indices of the paired features.

This additional penalty can be handled easily in the optimization framework. For the direct algorithm (Algorithm
\ref{alg:direct_alg}), 
we simply add a new row to $\tilde X$ and $\tilde \y$ for each pairwise constraint, while the one-at-a-time algorithm (Algorithm \ref{alg:one-at-a-time})
can be similarly modified.

\subsection{Modeling interactions between views}
\label{sec:interactions}
In our general objective function \eqref{eq:zero0}, we can capture interactions between features in the same view, by using methods such as random forests or boosting for the learners $f_X $ and $f_Z$.
However, this framework does not allow for interactions 
 between features in {\em different} views. Here is an  objective function to facilitate such interactions:
\begin{equation}
    {\rm min} \;  {\rm E}\Bigl[\frac{1}{2}  (\y-f_X(X)-f_Z(Z) -f_{XZ}(X,Z))^2 + \frac{\rho}{2}(f_X(X)-f_Z(Z))^2 +\frac{\rho}{2(1-\rho)} f^2_{XZ}(X,Z)\Bigr],
\label{eq:zero2}
\end{equation}
where $f_{XZ}(X,Z)$ is a joint function of $X$ and $Z$, including for example, interactions between the features in each view.

The solution to \eqref{eq:zero2} has fixed points:
\begin{eqnarray}
f_X(X)&=&{\rm E} \Bigl[\frac{\y}{1+\rho}-\frac{(1-\rho)f_Z(Z)}{(1+\rho)} -\frac{f_{XZ}(X,Z)}{1+\rho} |X\Bigr],\cr
f_Z(Z)&=&{\rm E} \Bigr[\frac{\y}{1+\rho}-\frac{(1-\rho)f_X(X)}{(1+\rho)} - \frac{f_{XZ}(X,Z)}{1+\rho} |Z\Bigr],\cr
f_{XZ}(X,Z)&=&{\rm E}  \Bigr[(1-\rho)(\y-f_X(X)-f_Z(Z) )| X,Z\Bigr].
\end{eqnarray}
When $\rho=0$,  from \eqref{eq:zero2} the solution reduces to the additive model $f_X(X)+f_Z(Z)+f_{XZ}(X,Z)$. As $\rho \rightarrow 1$,  the joint 
term $f_{XY} \rightarrow 0$ and we again get the late fusion estimate as the average of the marginal predictions $\hat f_X(X)$ and $\hat f_Z(Z)$.
To implement this in practice, we simply insert learners such as random forest or boosting for $f_X, f_Z$ and $f_{XZ}$.
The first two use only features from $X$ and $Z$, while the last uses features from both.

\section{Discussion}
\label{sec:discussion}
In this paper, we introduce a new method called cooperative learning for supervised learning with multiple set of features, or ``data views''.
The method encourages the predictions from different data views to align through an agreement penalty.
By varying the weight of the agreement penalty in the objective, we obtain a spectrum of solutions that include the commonly-used early and late fusion methods.
The method can choose the degree of agreement (or fusion) in an data-adaptive manner.
Cooperative learning provides a powerful tool for multiomics data fusion by strengthening aligned signals across modalities and allowing flexible fitting mechanisms for different modalities.
The effectiveness of our methodology has implications for improving diagnostics and therapeutics in an increasingly multiomic world. 

Furthermore, cooperative learning could be extended to the semi-supervised setting when we have additional matched data views on samples that are unlabeled.
The agreement penalty allows us to leverage the signals in the matched unlabeled samples to our advantage.
In addition, when we have missing values in some data views, the agreement penalty also allows us to impute one view from the other(s).
Lastly, the method can be easily extended to binary, count and survival data. 
An open-source R language package for cooperative learning called  \href{https://cran.r-project.org/web/packages/multiview/index.html}{{\tt multiview}} is available on the CRAN repository.

\medskip

{\bf Acknowledgments}. We would like to thank  Olivier Gevaert, Trevor Hastie, Ryan Tibshirani, and Samson Mataraso for helpful discussions, and two referees
whose comments greatly improved this manuscript.
D.Y.D was supported by the Stanford Graduate Fellowship (SGF).
B.N. was supported by Stanford Clinical \& Translational Science Award grant 5UL1TR003142-02 from the NIH National Center for Advancing Translational Sciences (NCATS).
R.T. was supported by the National
Institutes of Health (5R01 EB001988-16) and the National Science Foundation (19 DMS1208164).

\bibliography{reference} 
\bibliographystyle{agsm}
\newpage
\begin{appendix}

\section{Adaptive cooperative learning}
\label{sec:adap_coop}

In this section, we outline an adaptive strategy for optimizing over $\lambda_x$ and $\lambda_z$ for different data views. 
We call this {\em adaptive cooperative learning}.
The method incorporates the values of $\lambda_{x}$ and $\lambda_z$ that have been adaptively optimized by the one-at-a-time algorithm (Algorithm \ref{alg:adaptive_iterative}) as a penalty factor in the direct algorithm (Algorithm \ref{alg:adaptive_full}).
In the two-dimensional grid of $\lambda_{x}$ and $\lambda_{z}$, our proposed strategy works by iteratively searching along one axis of $\lambda$ while fixing the other constant.
\vspace{6mm}

\begin{algorithm}[H]
\KwIn{$X \in \mR^{n\times p_x}$ and $Z \in
\mR^{n\times p_z}$, the response $\vec{y} \in \mR^{n}$, and a fixed hyperparameter $\rho \in \mR$.}
\KwOut{$\hat{\tx}$ and $\hat{\tz}$ from the last iteration, along with the hyperparameters $\lambda_x^{*}$ and $\lambda_z^{*}$ and the corresponding CV errors.}
    \begin{enumerate}
        \item Initialize $\ttheta_x^{(0)} \in \mR^{p_x}$ and $\ttheta_z^{(0)} \in \mR^{p_z}$.
        \item For $k = 0, 1, 2, \ldots$ until convergence:
        \begin{enumerate}
            \item Set $\vec{y_x^*}=
\frac{\y}{1+\rho}-\frac{(1-\rho)Z\tz}{(1+\rho)}$. Solve  ${\rm Lasso}(X, \y_x^*,\lambda)$ over a decreasing grid of $\lambda$ values. Update $\ttheta_x^{(k+1)}$ to be the solution and record the hyperparameter $\lambda_{x}^{*}$ that minimizes the CV error.
            \item Set $\vec{y_z^*}=
\frac{\y}{1+\rho}-\frac{(1-\rho)X\tx}{(1+\rho)}$. Solve  ${\rm Lasso}(Z, \y_z^*,\lambda)$ over a decreasing grid of $\lambda$ values. Update $\ttheta_z^{(k+1)}$ to be the solution and record the hyperparameter $\lambda_{z}^{*}$ that minimizes the CV error.
        \end{enumerate}
    \end{enumerate}
\caption{\em One-at-a-time algorithm for adaptive cooperative learning (regression).}
\label{alg:adaptive_iterative}
\end{algorithm}

\begin{algorithm}[H]
\KwIn{$X \in \mR^{n\times p_x}$ and $Z \in
\mR^{n\times p_z}$, the response $\vec{y} \in \mR^{n}$, and a grid of hyperparameter values ($\rho_{\tt min}, \ldots, \rho_{\tt max}).$}
\vspace{2mm}

\For{$\rho \gets \rho_{\tt min},\hdots, \rho_{\tt max}$}{
        Run {\tt Algorithm \ref{alg:adaptive_iterative}} with both (X,Z) and (Z,X) with the same folds for CV. Select the one with the lower sum of the two CV errors. Get the corresponding $\lambda_x^{*}$ and $\lambda_z^{*}$.\\
        
        Set \begin{align*} \tilde X=
\begin{pmatrix}
 X   & Z\\
-\sqrt{\rho}X &  \sqrt{\rho}Z
\end{pmatrix}, \tilde \y=
\begin{pmatrix}
\y \\
\vec{0}
\end{pmatrix}.
\end{align*}
Solve ${\rm Lasso}(\tilde X, \tilde \y, \lambda)$ over a decreasing grid of $\lambda$ values, with a penalty factor of 
$(1, \ldots, 1, \frac{\lambda_z^{*}}{\lambda_x^{*}},  \ldots, \frac{\lambda_z^{*}}{\lambda_x^{*}})$. Note that we form folds from the rows of X and Z and then construct the corresponding $\tilde X$. 
}

Select the optimal value of $\rho$ based on the CV error and get the final fit.

\caption{\em Direct algorithm for adaptive cooperative learning (regression).}\label{alg:adaptive_full}
\end{algorithm}

\section{More comprehensive simulation studies on cooperative regularized regression}
\label{sec:appendix_sim}

\subsection{More simulation results of the high-dimensional settings ($p=1000, n=200$)}

\begin{figure}[h!]  
\includegraphics[width=1\textwidth]{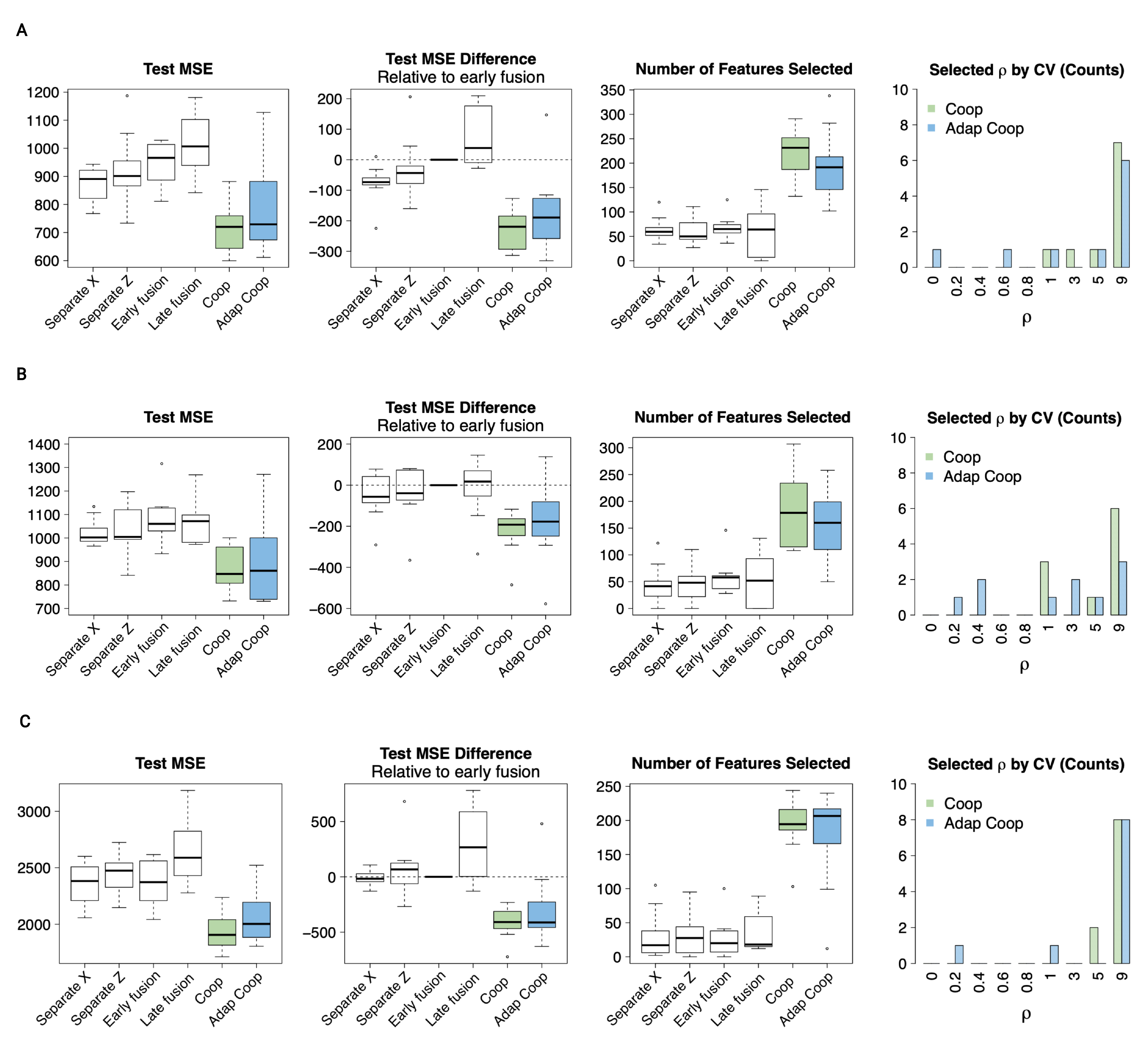}
\caption{{\em Simulation studies on cooperative regularized linear regression when X and Z are high-dimensional and have a high level of correlation with each other.} {\em (A)} Simulation results when $X$ and $Z$ have a high level of correlation and both contain signal ($t_{x} = t_{z} = 6$), $n = 200, p = 1000$, SNR $= 1.0$. The first panel shows MSE on a test set; the second panel shows the MSE difference on the test set relative to early fusion; the third panel shows the number of features selected; the fourth panel shows the $\rho$ values selected by CV in cooperative 
learning. Here ``Coop'' refers to cooperative 
learning outlined in Algorithm 1 and ``Adap Coop'' refers to adaptive cooperative learning outlined in Algorithm \ref{alg:adaptive_full}. 
{\em (B)} Simulation results when $X$ and $Z$ have a high level of correlation and both contain signal ($t_{x} = t_{z} = 6$), $n = 200, p = 1000$, SNR $= 0.6$.
{\em (C)} Simulation results when $X$ and $Z$ have a high level of correlation, X contains more signal than Z ($t_{x} = 4, t_{z} = 2$), $n = 200, p = 1000$, SNR $= 0.6$.
}
\label{fig:res_1}
\end{figure}

\begin{figure}[h!]  
\includegraphics[width=1\textwidth]{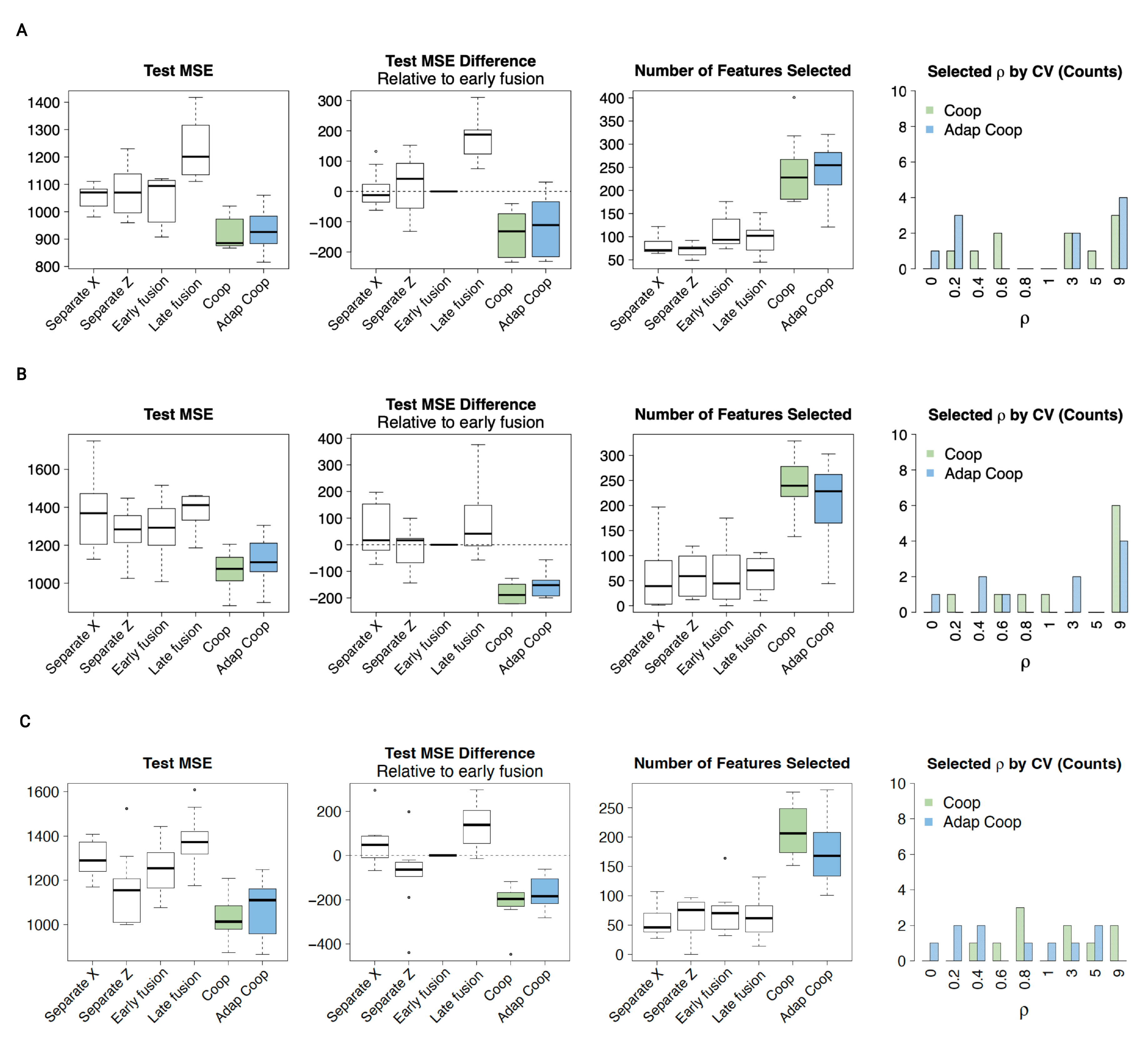}
\caption{{\em Simulation studies on cooperative regularized linear regression when X and Z are high-dimensional and have a medium level of correlation with each other.} {\em (A)}  Simulation results when $X$ and $Z$ have a medium level of correlation and both contain signal ($t_{x} = t_{z} = 2$), $n = 200, p = 1000$, SNR $= 3.5$. The setup is the same as in Figure \ref{fig:res_1}.
{\em (B)} Simulation results when $X$ and $Z$ have a medium level of correlation and both contain signal ($t_{x} = t_{z} = 2$), $n = 200, p = 1000$, SNR $= 1.6$.
{\em (C)} Simulation results when $X$ and $Z$ have a medium level of correlation, and Z contains more signal than X ($t_{x} = 2, t_{z} = 3$), $n = 200, p = 1000$, SNR $= 1.5$.}
\end{figure}

\begin{figure}[h!]  
\includegraphics[width=1\textwidth]{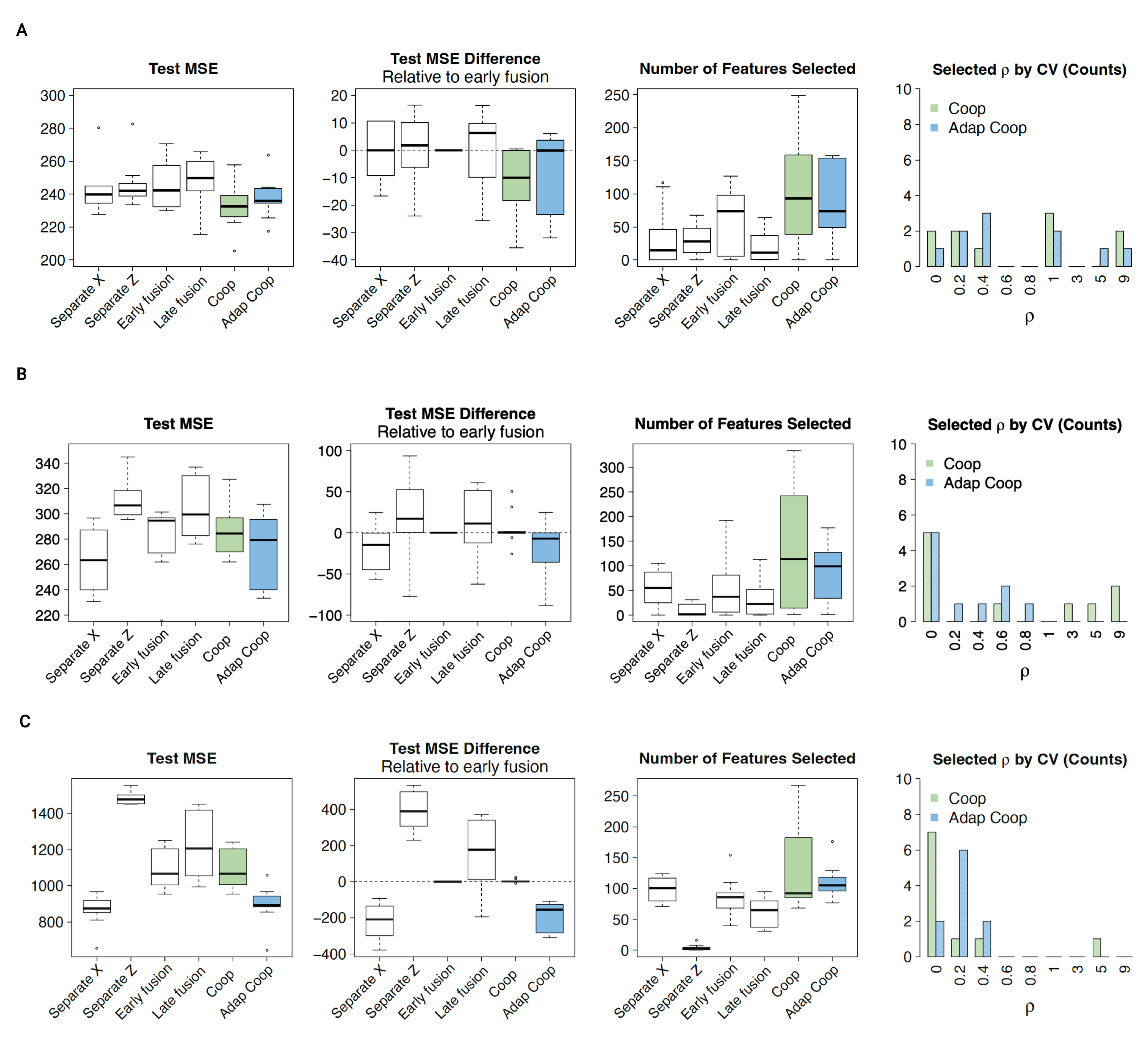}
\caption{{\em Simulation studies on cooperative regularized linear regression when X and Z are high-dimensional and have no correlation.} {\em (A)}   Simulation results when $X$ and $Z$ have no correlation, and both $X$ and $Z$ contain signal (here we generated $y$ as a linear combination of $X$ and $Z$ instead of the latent factors), $n = 200, p = 1000$, SNR $= 1.0$. The setup is the same as in Figure \ref{fig:res_1}.
{\em (B)} Simulation results when $X$ and $Z$ have no correlation; $X$ contains more signal than $Z$ (here we generated $y$ as a linear combination of $X$ and $Z$ instead of the latent factors), $n = 200, p = 1000$, SNR $= 1.1$.
{\em (C)} Simulation results when $X$ and $Z$ have no correlation; only $X$ contains signal ($t_{x} = 2, t_{z} = 0$), $n = 200, p = 1000$, SNR $= 3.5$.}
\end{figure}
\clearpage

\subsection{Simulation results of the lower-dimensional settings ($p=200, n=500$)}

\begin{figure}[h!]  
\includegraphics[width=1\textwidth]{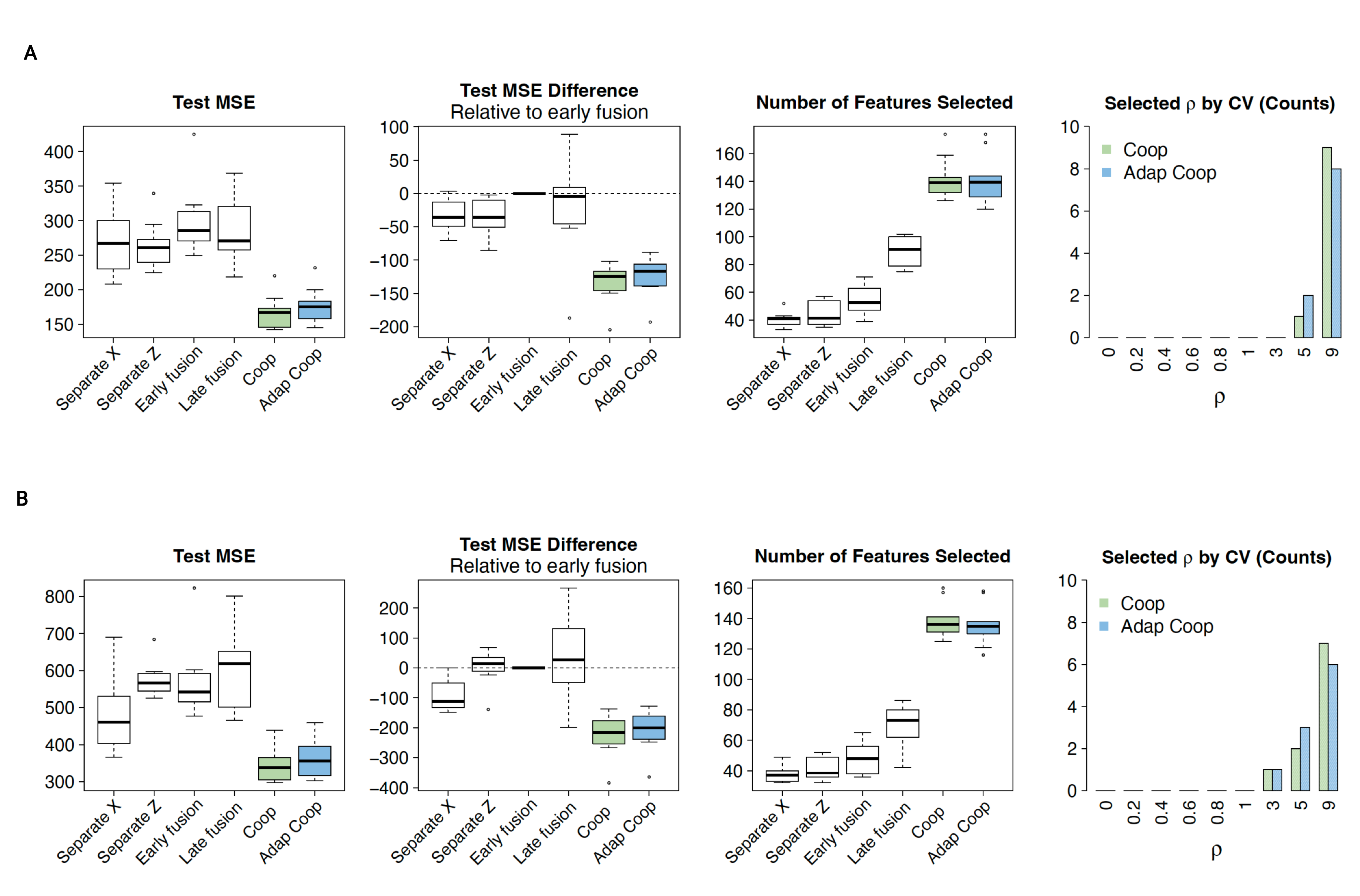}
\caption{{\em Simulation studies on cooperative regularized linear regression when X and Z are of a lower dimension and have a high level of correlation with each other.} {\em (A)} Simulation results when $X$ and $Z$ have a high level of correlation and both contain signal ($t_{x} = t_{z} = 6$), $n = 500, p = 200$, SNR $= 1.2$. The first panel shows MSE on a test set; the second panel shows the MSE difference on the test set relative to early fusion; the third panel shows the number of features selected; the fourth panel shows the $\rho$ values selected by CV in cooperative 
learning. Here ``Coop'' refers to cooperative 
learning outlined in Algorithm 1 and ``Adap Coop'' refers to adaptive cooperative learning outlined in Algorithm \ref{alg:adaptive_full}. 
{\em (B)} Simulation results when $X$ and $Z$ have a high level of correlation and X contains more signal than Z ($t_{x} = 5, t_{z} = 3$), $n = 500, p = 200$, SNR $= 0.7$.}
\label{fig:res_S4}
\end{figure}
\clearpage

\begin{figure}[h!]  
\includegraphics[width=1\textwidth]{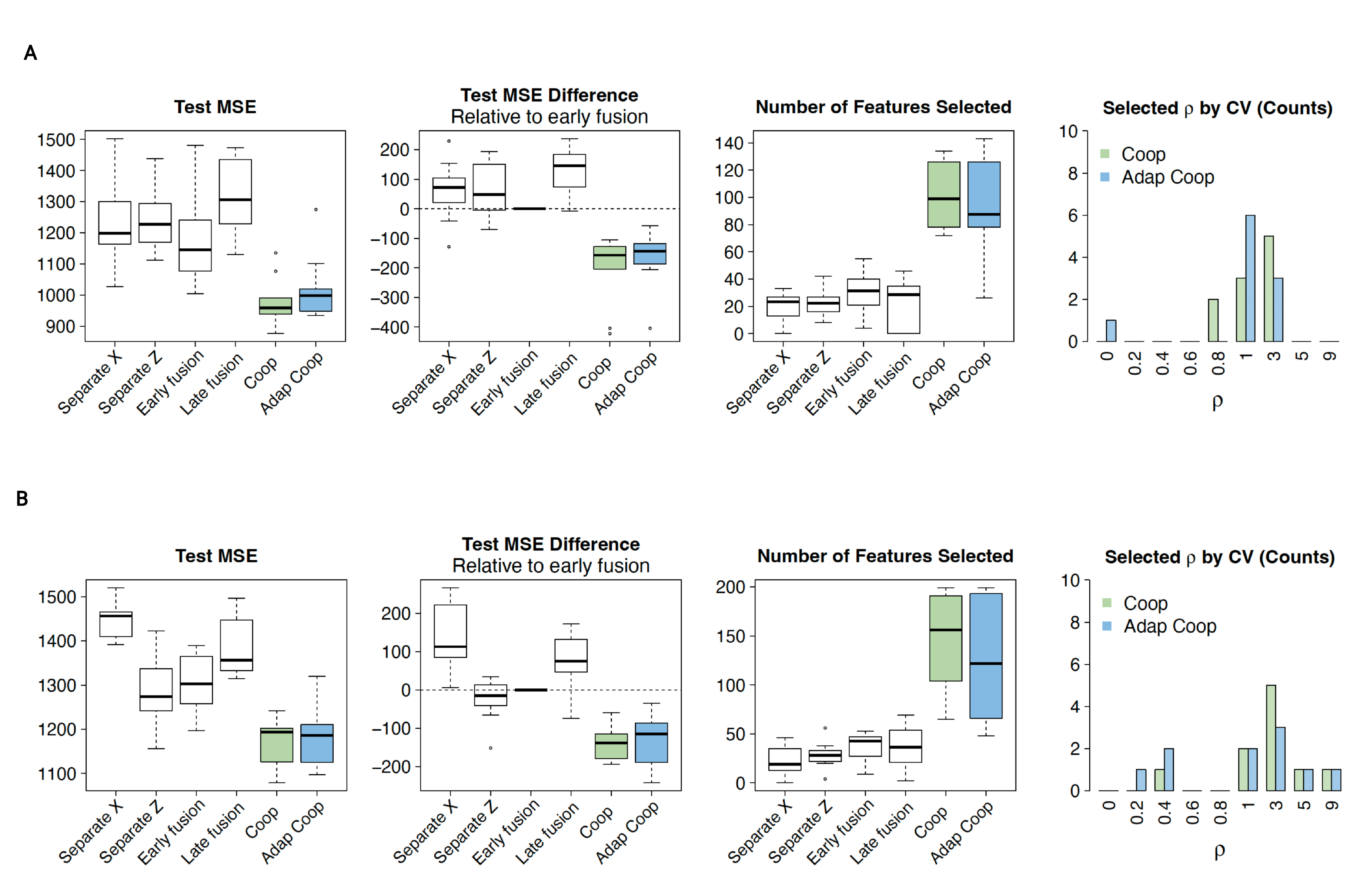}
\caption{{\em Simulation studies on cooperative regularized linear regression when X and Z are of a lower dimension and have a medium level of correlation with each other.} {\em (A)} Simulation results when $X$ and $Z$ have a medium level of correlation and both contain signal ($t_{x} = t_{z} = 1$), $n = 500, p = 200$, SNR $= 0.8$. The setup is the same as in Figure \ref{fig:res_S4}.
{\em (B)} Simulation results when $X$ and $Z$ have a medium level of correlation, and $Z$ contains more signal than $X$ ($t_{x} = 0.6, t_{z} = 0.9$), $n = 500, p = 200$, SNR $= 0.5$.}
\end{figure}
\clearpage

\begin{figure}[h!]  
\includegraphics[width=1\textwidth]{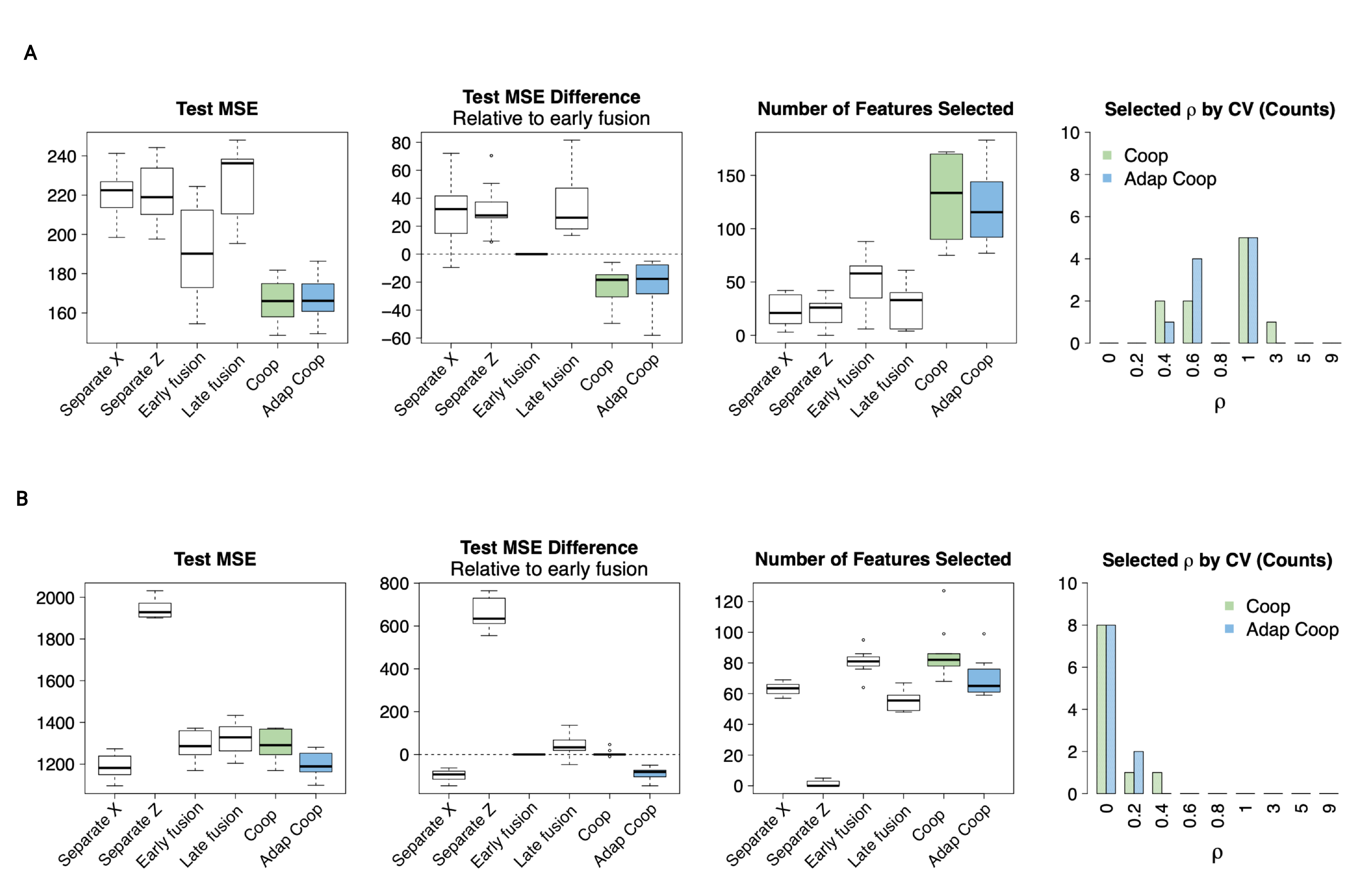}
\caption{{\em Simulation studies on cooperative regularized linear regression when X and Z are of a lower dimension and have no correlation with each other.} {\em (A)} Simulation results when $X$ and $Z$ have no correlation ($t_{x} = t_{z} = 0$), and both $X$ and $Z$ contain signal (here we generated $y$ as a linear combination of $X$ and $Z$ instead of the latent factors), $n = 500, p = 200$, SNR $= 0.3$. The setup is the same as in Figure \ref{fig:res_S4}.
{\em (B)} Simulation results when $X$ and $Z$ have no correlation and only $X$ contains signal ($t_{x} = 2, t_{z} = 0$), $n = 500, p = 200$, SNR $= 3.0$.}
\end{figure}

\section{Simulation studies on cooperative learning for more than two data views}
\label{sec:appendix_more_than_two_views}

Here we conduct simulation studies on cooperative learning for more than two data views. 
Specifically, we consider the setting of three data views, and this generalizes easily to more data views.
We generated Gaussian data with $n=200$ and $p=300$ in each of the views $X_1$,  $X_2$ and $X_3$, and created correlation between them using latent factors.
The response $\y$ was generated as a linear combination of the latent factors, corrupted by Gaussian noise.

\subsection{Simulation procedure for more than two data views}
The simulation for 3 data views is set up as follows.

Given values for parameters $n, p_{x_{1}}, p_{x_{2}}$, $p_{x_{3}}, p_{u}, s_{u}, t_{x_{1}}, t_{x_{2}}, t_{x_{3}}, \bbeta_{u}, \sigma$, we generate data according to the following procedure: 
\begin{enumerate}
    \item $x_{1j} \in \mR^{n}$ distributed i.i.d. MVN$(0, I_{n})$ for $j = 1,2,\ldots,p_{x_{1}}$.
    \item $x_{2j} \in \mR^{n}$ distributed i.i.d. MVN$(0, I_{n})$ for $j = 1,2,\ldots,p_{x_{2}}$.
    \item $x_{3j} \in \mR^{n}$ distributed i.i.d. MVN$(0, I_{n})$ for $j = 1,2,\ldots,p_{x_{3}}$.
    \item For $i = 1, 2, \ldots, p_{u}$ ($p_{u}$ corresponds to the number of latent factors):
    \begin{enumerate}
        \item $u_{i} \in \mR^{n}$ distributed i.i.d. MVN$(0, s_{u}^2I_{n})$;
        \item $x_{1i} = x_{1i} + t_{x_{1}} * u_{i}$;
        \item $x_{2i} = x_{2i} + t_{x_{2}} * u_{i}$;
        \item $x_{3i} = x_{3i} + t_{x_{3}} * u_{i}$.
    \end{enumerate}
    \item $X_{1}=[x_{11}, x_{12}, \ldots, x_{1p_{x_{1}}}]$, $X_{2}=[x_{21}, x_{22}, \ldots, x_{2p_{x_{2}}}]$, $X_{3}=[x_{31}, x_{32}, \ldots, x_{3p_{x_{3}}}]$.
    \item $U = [u_{1}, u_{2}, \ldots, u_{p_{u}}]$, $\y = U\bbeta_{u} + \epsilon$ where $\epsilon \in \mR^{n}$ distributed i.i.d. MVN$(0, \sigma^2 I_{n})$.
\end{enumerate}

We compare the following methods: (1) separate $X_1$, separate $X_2$, and separate $X_3$: the standard lasso is applied on the separate data views of $X_1$,  $X_2$ and $X_3$ with 10-fold CV; (2) early fusion: the standard lasso is applied on the concatenated data views of $X_1$,  $X_2$ and $X_3$ with 10-fold CV (note that this is equivalent to cooperative learning with $\rho = 0$); (3) late fusion: separate lasso models are first fitted on $X_1$,  $X_2$ and $X_3$ independently with 10-fold CV, and the three resulting predictors are then combined through linear least squares for the final prediction; (4) cooperative learning (regression) and adaptive cooperative learning.

We evaluated the performance based on the mean-squared error (MSE) on a test set and conducted each simulation experiment 10 times.

\subsection{Simulation results for more than two data views}

Figure \ref{fig:res_more_than_2_data_views_1} and \ref{fig:res_more_than_2_data_views_2} show the simulation results for 3 data views.
Overall, the simulation results can be summarized as follows:
\begin{itemize}
\item Cooperative learning performs the best in terms of test MSE across the range of SNR and correlation settings.
It is most helpful when the data views are correlated and contain signal (as in Figure \ref{fig:res_more_than_2_data_views_1}{\em A} and Figure \ref{fig:res_more_than_2_data_views_2}{\em A}). 
When the correlation between data views is higher, higher values of $\rho$ are more likely to be selected.

\item When only two data views are correlated and contain signal (as in Figure \ref{fig:res_more_than_2_data_views_1}{\em B} and Figure \ref{fig:res_more_than_2_data_views_2}{\em C}), cooperative learning also gives performance gains by leveraging the correlation through the agreement penalty, while early fusion can be outperformed by the separate models fit on the data views containing the signal.

\item When only one view contains signal and the views are not correlated (as in Figure \ref{fig:res_more_than_2_data_views_1}{\em C}), cooperative learning is outperformed by the separate model fit on the view containing the signal, but adaptive cooperative learning is able to perform on par with the separate model, outperforming early and late fusion.
\end{itemize}

\begin{figure}[h!]
\includegraphics[width=1\textwidth]{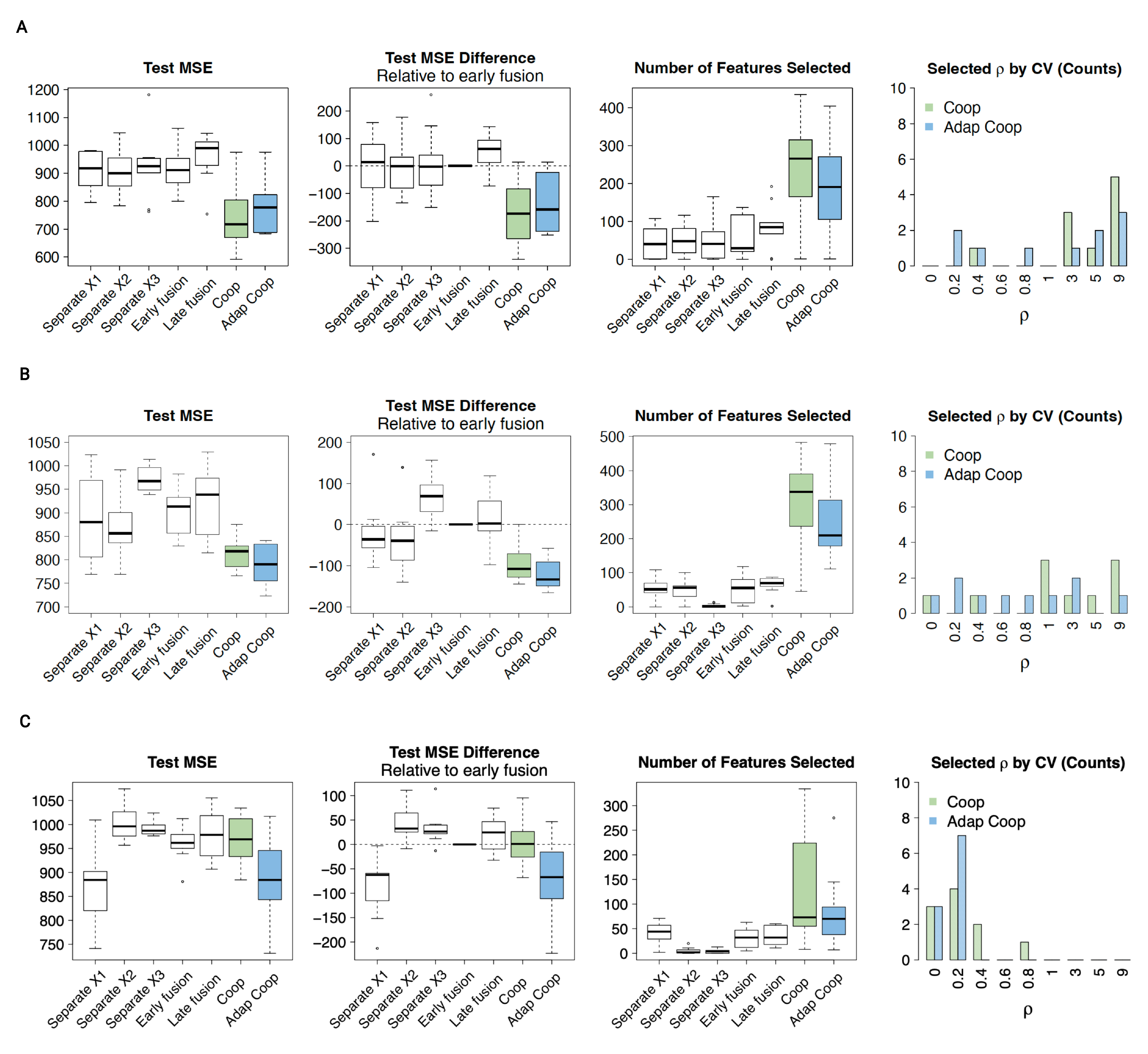}
\caption{{\em Simulation studies on cooperative regularized linear regression for more than two data views.} {\em (A)} Simulation results when $X_1$, $X_2$ and $X_3$ are correlated and all contain signal ($t_{x_1} = t_{x_2} = t_{x_3} = 2$), $n = 200, p = 900$, SNR $= 1.5$. The first panel shows MSE on a test set; the second panel shows the MSE difference on the test set relative to early fusion; the third panel shows the number of features selected; the fourth panel shows the $\rho$ values selected by CV in cooperative 
learning. Here ``Coop'' refers to cooperative 
learning outlined in Algorithm 1 and ``Adap Coop'' refers to adaptive cooperative learning outlined in Algorithm \ref{alg:adaptive_full}. 
{\em (B)} Simulation results when only $X_1$ and $X_2$ are correlated and contain signal ($t_{x_1} = t_{x_2} = 2$, $t_{x_3} = 0$), $n = 200, p = 900$, SNR $= 1.5$.
{\em (C)} Simulation results when $X_1$, $X_2$ and $X_3$ are uncorrelated, and only $X_1$ contains signal ($t_{x_1} = 2, t_{x_2} = t_{x_3} = 0$), $n = 200, p = 900$, SNR $= 1.5$.
}
\label{fig:res_more_than_2_data_views_1}
\end{figure}

\begin{figure}[h!]
\includegraphics[width=1\textwidth]{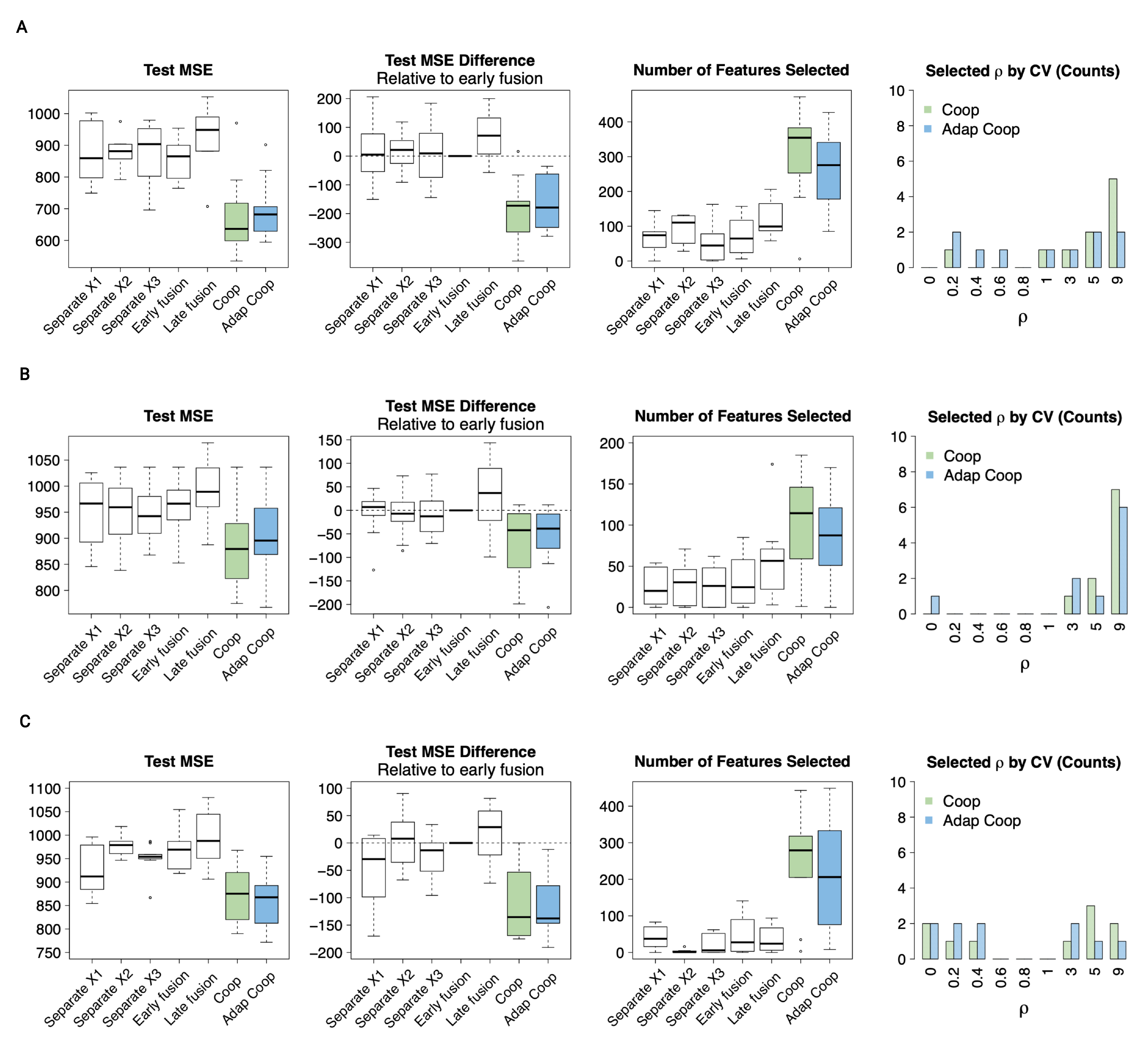}
\caption{{\em Simulation studies on cooperative regularized linear regression for more than two data views.} {\em (A)} Simulation results when $X_1$, $X_2$ and $X_3$ are correlated and all contain signal ($t_{x_1} = t_{x_2} = t_{x_3} = 2$), $n = 200, p = 900$, SNR $= 2.5$. The first panel shows MSE on a test set; the second panel shows the MSE difference on the test set relative to early fusion; the third panel shows the number of features selected; the fourth panel shows the $\rho$ values selected by CV in cooperative 
learning. Here ``Coop'' refers to cooperative 
learning outlined in Algorithm 1 and ``Adap Coop'' refers to adaptive cooperative learning outlined in Algorithm \ref{alg:adaptive_full}. 
{\em (B)}  Simulation results when $X_1$, $X_2$ and $X_3$ are correlated and all contain signal ($t_{x_1} = t_{x_2} = t_{x_3} = 2$), $n = 200, p = 900$, SNR $= 0.6$.
{\em (c)} Simulation results when only $X_1$ and $X_3$ are correlated; $X_1$ contains more signal than $X_3$, $X_2$ does not contain signal ($t_{x_1} = 2$, $t_{x_2} = 0$, $t_{x_3} = 1.5$), $n = 200, p = 900$, SNR $= 1.0$.
}
\label{fig:res_more_than_2_data_views_2}
\end{figure}
\clearpage

\section{Theoretical analysis under the factor model}
\label{sec:appendix_theoretical_analysis}

To understand the role of the agreement penalty from a theoretical perspective, we consider the following latent factor model. Let $\bu = (U_1, U_2, \dots, U_n)$ be a vector of $n$ i.i.d. random variables with $U_i \sim \mathcal{N}(0,1)$. Let $\by = (Y_1, \dots, Y_n)$, $\bx = (X_1, \dots, X_n)$, and $\bz = (Z_1, \dots, Z_n)$, with
\begin{equation}
\label{eqn:latent_model}
    Y_i = \gy U_i + \varepsilon_{yi}, \qquad 
    X_i = \gx U_i + \varepsilon_{xi}, \qquad \text{and} \qquad
    Z_i = \gz U_i + \varepsilon_{zi},
\end{equation}
where $\varepsilon_{yi} \sim \mathcal{N}\p{0,\sigma_y^2}$, $\varepsilon_{xi} \sim \mathcal{N}\p{0,\sigma_x^2}$, $\varepsilon_{zi} \sim \mathcal{N}\p{0,\sigma_z^2}$ independently. 

In this section, we study the mean squared error (MSE) of the cooperative learning algorithm. More precisely, let
\begin{equation}
\label{eqn:theta_hat}
    \hat{\theta} = \operatorname{argmin}_{\theta}\sum_{i = 1}^n \sqb{\frac{1}{2}\p{Y_i - X_i\theta_x - Z_i \theta_z}^2 + \frac{\rho}{2}\p{X_i\theta_x - Z_i \theta_z}^2}. 
\end{equation}
Let $U\new$, $X\new$, $Y\new$, $Z\new$ be some new random variables generated from \eqref{eqn:latent_model} independently of the previous data, i.e.,
\begin{equation}
        Y\new = \gy U\new +\varepsilon_{y \operatorname{new}}, \qquad 
        X\new = \gx U\new +\varepsilon_{x \operatorname{new}}, \qquad \text{and} \qquad
        Z\new = \gz U\new +\varepsilon_{z \operatorname{new}},
\end{equation}
where $U\new \sim \mathcal{N}(0,1)$, $\varepsilon_{y \operatorname{new}} \sim \mathcal{N}\p{0,\sigma_y^2}$, $\varepsilon_{x \operatorname{new}} \sim \mathcal{N}\p{0,\sigma_x^2}$, $\varepsilon_{z \operatorname{new}} \sim \mathcal{N}\p{0,\sigma_z^2}$ independently. We focus on the MSE conditioning on $\bx$ and $\bz$:
\begin{equation}
\label{eqn:mse_def}
    \operatorname{MSE}(\bx, \bz; \rho) = \EE{\p{Y\new - \p{X\new \hat{\theta}_x + Z\new \hat{\theta}_z}}^2 \mid \bx, \bz}. 
\end{equation}
The case of $\rho = 0$ corresponds to the linear regression with no agreement penalty. We will study the behavior of $\operatorname{MSE}(\bx, \bz; \rho)$ when $\rho$ is around $0$. 

\begin{prop}
\label{prop:dev_neg}
The derivative of $\operatorname{MSE}(\bx, \bz; \rho)$ satisfies
\begin{equation}
\label{eqn:neg_cond}
\frac{d}{d\rho}\sqb{\operatorname{MSE}(\bx, \bz; \rho)} |_{\rho = 0} = \sigma^{\star 2} (C_2 B_1 - 2 C_1 B_2)/C_2^3,
\end{equation}
where 
\begin{equation}
\label{eqn:C1_B2}
\begin{split}
\sigma^{\star 2} & = \frac{\gy^2}{1 + \gamma_x^2/\sigma_x^2 + \gamma_z^2/\sigma_z^2} + \sigma_y^2,\\
C_1 &= \sqb{(\gx^2 + \sigma_x^2)(\bz \trans \bz) + (\gz^2 + \sigma_z^2)(\bx \trans \bx) - 2\gx \gz (\bx \trans \bz)}((\bx\trans \bx)(\bz\trans \bz) - (\bx\trans \bz)^2),\\
B_1 & = 2\sqb{(\gx^2 + \sigma_x^2)(\bz \trans \bz) + (\gz^2 + \sigma_z^2)(\bx \trans \bx) + 2\gx \gz (\bx \trans \bz)}((\bx\trans \bx)(\bz\trans \bz) - (\bx\trans \bz)^2),\\
C_2 &= (\bx\trans \bx)(\bz\trans \bz) - (\bx\trans \bz)^2,\\
B_2 &= 2\p{(\bx\trans \bx)(\bz\trans \bz) + (\bx\trans \bz)^2}. 
\end{split}
\end{equation}
\end{prop}

\vspace{3mm}

\begin{prop}
\label{prop:dev_big_O}
The derivative of $\operatorname{MSE}(\bx, \bz; \rho)$ at $\rho = 0$ satisfies 
\begin{equation}
\frac{d}{d\rho}\sqb{\operatorname{MSE}(\bx, \bz; \rho)} |_{\rho = 0} = -\frac{4}{n}\p{1 + \frac{2 \gx^2 \gz^2}{ \sigma_x^2\gz^2 + \sigma_z^2 \gx^2 + \sigma_x^2 \sigma_z^2 }}\p{\sigma_y^2 + \frac{\gy^2 \sigma_x^2 \sigma_z^2}{\sigma_x^2\gz^2 + \sigma_z^2 \gx^2 + \sigma_x^2 \sigma_z^2} } + \mathcal{O}_p\p{n^{-\frac{3}2}}. 
\end{equation}
Here the notation $\mathcal{O}_p(\cdot)$ is used with the following meaning: 
$X_{n}= \mathcal{O}_{p}\left(a_{n}\right) \text { as } n \rightarrow \infty$
means that for any $\varepsilon>0$, there exists a finite $M>0$ and a finite $N>0$ such that $
\PP{\left|X_{n} / a_{n}\right|>M}<\varepsilon, \forall n>N$. 
\end{prop}

The proposition establishes that the MSE is a decreasing function of $\rho$ around 0 with high probability. In other words, the agreement penalty is helpful in reducing the mean squared error. 
To further interpret the above results, we study the ratio of the derivative to the MSE itself. 

\vspace{3mm}

\begin{prop}
\label{prop:dev_ex_MSE}
The ratio of the derivative of $\operatorname{MSE}(\bx, \bz; \rho)$ at $\rho = 0$ and ${\operatorname{MSE}(\bx, \bz; 0)}$ satisfies 
\begin{equation}
\frac{\frac{d}{d\rho}\sqb{\operatorname{MSE}(\bx, \bz; \rho)} |_{\rho = 0}}{\operatorname{MSE}(\bx, \bz; 0)} 
= -\frac{4}{n}\p{1 + \frac{2 \gx^2 \gz^2}{ \sigma_x^2\gz^2 + \sigma_z^2 \gx^2 + \sigma_x^2 \sigma_z^2 }} + \mathcal{O}_p\p{n^{-\frac{3}{2}}}. 
\end{equation}
Here the notation $\mathcal{O}_p(\cdot)$ is used with the same meaning as in Proposition \ref{prop:dev_big_O}. 
\end{prop}

Proposition \ref{prop:dev_ex_MSE} presents a simple form of the ratio of the derivative to the MSE itself. 
The ratio quantifies by what percentage the ``agreement'' penalty decreases the MSE. 
It can been seen from this representation that this ratio depends on the structure of the factor model, and that the agreement penalty is more helpful when the sample size $n$ is smaller. In the extreme case, when we have infinite data, i.e., when $n = \infty$, the derivative of the MSE is 0; in this case, we learn all the signals from the data even without the agreement penalty.

\subsection{Proof of Proposition \ref{prop:dev_neg}}

Here we present a lemma that is used in the proof of Proposition \ref{prop:dev_neg}. 
\begin{lemm}
\label{lemm:var_der}
Assume that $c_2 > 0$. Let $g(\rho) = (a_1\rho^2 + b_1 \rho + c_1)/(a_2\rho^2 + b_2 \rho + c_2)^2$, then 
\begin{equation}
    g'(\rho)|_{\rho = 0} = \frac{b_1 c_2 - 2 b_2 c_1}{c_2^3}. 
\end{equation}
\end{lemm}
\begin{proof}
We compute the derivative of the function $g$:
\begin{equation}
    g'(\rho) = \frac{(2 a_1\rho + b_1)(a_2 \rho^2 + b_2 \rho + c_2) - 2(2 a_2\rho + b_2)(a_1 \rho^2 + b_1 \rho + c_1)}{(a_2\rho^2 + b_2 \rho + c_2)^3}.
\end{equation}
Evaluating at $\rho = 0$, we get 
\begin{equation}
    g'(\rho)|_{\rho = 0} = \frac{b_1 c_2 - 2 b_2 c_1}{c_2^3}. 
\end{equation}
\end{proof}

With this lemma, we are ready to prove the proposition. 
We start with writing down an explicit expression for the estimator $\hat\theta$. 
Let 
\begin{equation}
\tilde X=
\begin{pmatrix}
 \bx  & \bz \\
-\sqrt{\rho}\bx &  \sqrt{\rho}\bz
\end{pmatrix}
,\quad 
\tilde \by = 
\begin{pmatrix}
\by \\
\boldsymbol{0}
\end{pmatrix}.
\end{equation}
Then \eqref{eqn:theta_hat} implies that $\thetah = (\thetah_x, \thetah_z)\trans$ takes the following form:
\begin{equation}
    \thetah = \p{\tilde{X} \trans\tilde{X}}^{-1} \p{\tilde{X} \trans \tilde{\by}}. 
\end{equation}
In particular, 
\begin{equation}
    \tilde{X} \trans\tilde{X}
    =   \begin{pmatrix}
         (1+\rho)\bx\trans \bx  &  (1-\rho)\bx\trans \bz \\
         (1-\rho)\bx\trans \bz  &   (1+\rho)\bz\trans \bz
        \end{pmatrix},
\end{equation}
and 
\begin{equation}
    \tilde{X} \trans \tilde{\by}
    = (\bx \trans \by, \bz \trans \by)\trans. 
\end{equation}
Therefore,
\begin{equation}
\begin{split}
\thetah = \begin{pmatrix}
            \thetah_x\\
            \thetah_z
        \end{pmatrix} &= \begin{pmatrix}
         (1+\rho)\bx\trans \bx  &  (1-\rho)\bx\trans \bz \\
         (1-\rho)\bx\trans \bz  &   (1+\rho)\bz\trans \bz
        \end{pmatrix}^{-1}
        \begin{pmatrix}
            \bx \trans \by\\
            \bz \trans \by
        \end{pmatrix}
    = \frac{1}{\det}\begin{pmatrix}
         (1+\rho)\bz\trans \bz  &  -(1-\rho)\bx\trans \bz \\
         -(1-\rho)\bx\trans \bz  &   (1+\rho)\bx\trans \bx
        \end{pmatrix}
        \begin{pmatrix}
            \bx \trans \by\\
            \bz \trans \by
        \end{pmatrix}\\
    &= \frac{1}{\det}
    \begin{pmatrix}
            (1+\rho)(\bz\trans \bz)( \bx \trans \by) - (1-\rho)(\bx\trans \bz) (\bz \trans \by)\\
            (1+\rho)(\bx\trans \bx) (\bz \trans \by) - (1-\rho)(\bx\trans \bz) (\bx \trans \by)
        \end{pmatrix},
\end{split}
\end{equation}
where $\det = (1+\rho)^2 (\bx\trans \bx) (\bz\trans \bz) - (1-\rho)^2 (\bx\trans \bz)^2$.

We then move on to analyze the conditional distribution of $\bu$ and $\by$ on $\bx$ and $\bz$. By \eqref{eqn:latent_model}, we can write down a joint distribution of $(U_i, X_i, Z_i)$:
\begin{equation}
\begin{split}
        \begin{pmatrix}
            U_i\\
            X_i\\
            Z_i
        \end{pmatrix} &\sim  \mathcal{N}
    \begin{pmatrix}
        \begin{pmatrix}
            0\\
            0\\
            0
        \end{pmatrix}\!\!,&
        \begin{pmatrix}
            1 & \gx & \gz\\
            \gx & \gx^2 + \sigma_x^2 & \gx \gz\\
            \gz & \gx\gz & \gz^2 + \sigma_z^2
        \end{pmatrix}
    \end{pmatrix}.
\end{split}
\end{equation}
Using formulas from conditional distribution of multivariate gaussian, we get that
\begin{equation}
    U_i \mid X_i, Z_i \sim \mathcal{N}(\EE{U_i \mid X_i, Z_i}, \Var{U_i \mid X_i, Z_i}),
\end{equation}
where 
\begin{equation}
\begin{split}
    \EE{U_i \mid X_i, Z_i} &= 
        \begin{pmatrix}
            \gx & \gz
        \end{pmatrix}
        \begin{pmatrix}
            \gx^2 + \sigma_x^2 & \gx \gz\\
            \gx\gz & \gz^2 + \sigma_z^2
        \end{pmatrix}^{-1}
        \begin{pmatrix}
            X_i\\
            Z_i
        \end{pmatrix}\\
& = \frac{1}{\sigma_x^2\sigma_z^2 + \gx^2\sigma_z^2 + \gz^2 \sigma_x^2} 
    \begin{pmatrix}
            \gx & \gz
        \end{pmatrix}
        \begin{pmatrix}
            \gz^2 + \sigma_z^2 & -\gx \gz\\
            -\gx\gz & \gx^2 + \sigma_x^2
        \end{pmatrix}
        \begin{pmatrix}
            X_i\\
            Z_i
        \end{pmatrix}\\
& = \frac{\gx X_i}{\sigma_x^2 + \gx^2 + \gz^2 \sigma_x^2/\sigma_z^2}
+ \frac{\gz Z_i}{\sigma_z^2 + \gz^2 + \gx^2 \sigma_z^2/\sigma_x^2},
\end{split}
\end{equation}
and 
\begin{equation}
\begin{split}
    \Var{U_i \mid X_i, Z_i} &= 1 - 
        \begin{pmatrix}
            \gx & \gz
        \end{pmatrix}
        \begin{pmatrix}
            \gx^2 + \sigma_x^2 & \gx \gz\\
            \gx\gz & \gz^2 + \sigma_z^2
        \end{pmatrix}^{-1}
        \begin{pmatrix}
            \gx\\
            \gz
        \end{pmatrix}\\
& = 1 - \frac{1}{\sigma_x^2\sigma_z^2 + \gx^2\sigma_z^2 + \gz^2 \sigma_x^2}
    \begin{pmatrix}
            \gx & \gz
        \end{pmatrix}
        \begin{pmatrix}
            \gz^2 + \sigma_z^2 & -\gx \gz\\
            -\gx\gz & \gx^2 + \sigma_x^2
        \end{pmatrix}
        \begin{pmatrix}
            \gx\\
            \gz
        \end{pmatrix}\\
& = 1 - \frac{ \gx^2\sigma_z^2 + \gz^2 \sigma_x^2}{\sigma_x^2\sigma_z^2 + \gx^2\sigma_z^2 + \gz^2 \sigma_x^2}
 = \frac{1}{1 + \gamma_x^2/\sigma_x^2 + \gamma_z^2/\sigma_z^2}.
\end{split}
\end{equation}
Since $Y_i = \gy U_i + \varepsilon_{yi}$, the above analysis implies that
\begin{equation}
\label{eqn:y_cond_0}
    Y_i \mid X_i, Z_i \sim \mathcal{N}(\EE{Y_i \mid X_i, Z_i}, \Var{Y_i \mid X_i, Z_i}),
\end{equation}
where 
\begin{equation}
\label{eqn:y_cond_1}
    \EE{Y_i \mid X_i, Z_i} = \gy \EE{U_i \mid X_i, Z_i}
 = \frac{\gx \gy X_i}{\sigma_x^2 + \gx^2 + \gz^2 \sigma_x^2/\sigma_z^2}
+ \frac{\gz \gy Z_i}{\sigma_z^2 + \gz^2 + \gx^2 \sigma_z^2/\sigma_x^2},
\end{equation}
and 
\begin{equation}
\label{eqn:y_cond_2}
    \Var{Y_i \mid X_i, Z_i}
    = \gy^2\Var{U_i \mid X_i, Z_i} + \sigma_y^2
 = \frac{\gy^2}{1 + \gamma_x^2/\sigma_x^2 + \gamma_z^2/\sigma_z^2} + \sigma_y^2.
\end{equation}

Let 
\begin{equation}
\label{eqn:thetas_def_second}
    \thetas_x = \frac{\gx \gy}{\sigma_x^2 + \gx^2 + \gz^2 \sigma_x^2/\sigma_z^2},
    \qquad 
    \thetas_z = \frac{\gz \gy}{\sigma_z^2 + \gz^2 + \gx^2 \sigma_z^2/\sigma_x^2},
    \qquad
    \sigma^{\star2} = \frac{\gy^2}{1 + \gamma_x^2/\sigma_x^2 + \gamma_z^2/\sigma_z^2} + \sigma_y^2,
\end{equation}
then  the above shows that we can express $Y_i$ as
\begin{equation}
\label{eqn:Y_decomp}
    Y_i = \thetas_x X_i + \thetas_z Z_i + \varepsilon^{\star}_i, 
\end{equation}
where $\varepsilon^{\star}_i \indp (X_i, Z_i)$ and $\varepsilon^{\star}_i \sim \mathcal{N}(0, \sigma^{\star2})$. 
In words, $Y_i$ can be decomposed into two independent terms: a linear combination of $X_i$ and $Z_i$, and an error term independent of $(X_i, Z_i)$. 

With the above tools, we are ready to study the MSE. Using \eqref{eqn:Y_decomp}, we can write
\begin{equation}
\label{eqn:MSE_decomp}
\begin{split}
    \operatorname{MSE}(\bx, \bz; \rho) 
    &= \EE{\p{Y\new - \p{X\new \hat{\theta}_x + Z\new \hat{\theta}_z}}^2 \mid \bx, \bz}\\
    &= \EE{\p{\thetas_x X\new + \thetas_z Z\new + \varepsilon^{\star}\new - \p{X\new \hat{\theta}_x + Z\new \hat{\theta}_z}}^2 \mid \bx, \bz}\\
    & = \EE{\p{(\thetas_x - \hat{\theta}_x) X\new + (\thetas_z - \hat{\theta}_z) Z\new + {\varepsilon^{\star}\new} }^2 \mid \bx, \bz}\\
    & = \EE{\p{(\thetas_x - \hat{\theta}_x) X\new + (\thetas_z - \hat{\theta}_z) Z\new }^2 \mid \bx, \bz}  + \EE{{\varepsilon^{\star}\new }^2 \mid \bx, \bz}\\
    & = \EE{\p{(\thetas_x - \hat{\theta}_x) X\new + (\thetas_z - \hat{\theta}_z) Z\new }^2 \mid \bx, \bz}  + \sigma^{\star 2}.
\end{split}
\end{equation}
Here the cross terms vanish because $\varepsilon^{\star}\new \indp (X \new, Z \new)$. Since the new dataset is independent of the training dataset, we can further simply the above:
\begin{equation}
\begin{split}
    \operatorname{MSE}(\bx, \bz; \rho) 
&= \EE{\p{(\thetas_x - \hat{\theta}_x) X\new}^2\mid \bx, \bz}  + \EE{\p{(\thetas_z - \hat{\theta}_z) Z\new }^2 \mid \bx, \bz} \\
& \qquad \qquad + 2\EE{\p{\thetas_z - \hat{\theta}_z}\p{\hat{\theta}_x - \thetas_x } Z\new X\new \mid \bx, \bz} 
    + \sigma^{\star 2}\\
& = \EE{\p{\hat{\theta}_x - \thetas_x  }^2\mid \bx, \bz} \EE{X \new^2} + \EE{\p{\hat{\theta}_z - \thetas_z }^2 \mid \bx, \bz} \EE{Z \new^2}\\
& \qquad \qquad + 2\EE{\p{\hat{\theta}_z - \thetas_z}\p{\hat{\theta}_x - \thetas_x }  \mid \bx, \bz} \EE{Z\new X\new}
    + \sigma^{\star 2}\\
& = \EE{\p{\hat{\theta}_x - \thetas_x  }^2\mid \bx, \bz} (\gx^2 + \sigma_x^2) + \EE{\p{\hat{\theta}_z - \thetas_z }^2 \mid \bx, \bz} \p{\gz^2 + \sigma_z^2}\\
& \qquad \qquad + 2\EE{\p{\hat{\theta}_z - \thetas_z}\p{\hat{\theta}_x - \thetas_x }  \mid \bx, \bz} \p{\gx \gz} 
+ \sigma^{\star 2}. 
\end{split}
\end{equation}
We can then further decompose the terms into squared bias plus variance. 
\begin{equation}
\label{eqn:mse_decom}
\begin{split}
    \operatorname{MSE}(\bx, \bz; \rho) 
& = \EE{\hat{\theta}_x - \thetas_x  \mid \bx, \bz}^2 (\gx^2 + \sigma_x^2) + \Var{\hat{\theta}_x \mid \bx, \bz} (\gx^2 + \sigma_x^2)\\
& \qquad \qquad + \EE{\hat{\theta}_z - \thetas_z \mid \bx, \bz}^2 \p{\gz^2 + \sigma_z^2} + \Var{\hat{\theta}_z \mid \bx, \bz} \p{\gz^2 + \sigma_z^2}\\
& \qquad \qquad + 2\EE{\hat{\theta}_z - \thetas_z\mid \bx, \bz} \EE{\hat{\theta}_x - \thetas_x   \mid \bx, \bz} \p{\gx \gz} 
+ 2\Cov{\hat{\theta}_z, \hat{\theta}_x \mid \bx, \bz}  \p{\gx \gz}
+ \sigma^{\star 2}\\
& = B^2(\bx, \bz; \rho) + V(\bx, \bz; \rho) + \sigma^{\star 2},
\end{split}
\end{equation}

\noindent where $B^2(\bx, \bz; \rho) = \EE{\hat{\theta}_x - \thetas_x  \mid \bx, \bz}^2 (\gx^2 + \sigma_x^2) + \EE{\hat{\theta}_z - \thetas_z \mid \bx, \bz}^2 \p{\gz^2 + \sigma_z^2} + 2\EE{\hat{\theta}_z - \thetas_z\mid \bx, \bz} \EE{\hat{\theta}_x - \thetas_x   \mid \bx, \bz} \p{\gx \gz}$ 
and
$V(\bx, \bz; \rho) = \Var{\hat{\theta}_x \mid \bx, \bz} (\gx^2 + \sigma_x^2) + \Var{\hat{\theta}_z \mid \bx, \bz} \p{\gz^2 + \sigma_z^2} + 2\Cov{\hat{\theta}_z, \hat{\theta}_x \mid \bx, \bz}  \p{\gx \gz} $ are the sum of bias related terms and the sum of variance related terms, respectively.

We can then use \eqref{eqn:y_cond_0} - \eqref{eqn:y_cond_2} to study the bias and variance of the estimators $\thetah_x$ and $\thetah_z$. We start with the bias. By \eqref{eqn:Y_decomp}, we have $\EE{\by \mid \bx, \bz} = \thetas_x \bx + \thetas_z \bz$. Therefore,
\begin{equation}
\begin{split}
\EE{\thetah_x \mid \bx, \bz} &=     
\EE{\frac{1}{\det}\p{ (1+\rho)(\bz\trans \bz)( \bx \trans \by) - (1-\rho)(\bx\trans \bz) (\bz \trans \by)} \mid \bx, \bz}\\
& = \frac{1}{\det}\p{ (1+\rho)(\bz\trans \bz)( \bx \trans \EE{\by\mid \bx, \bz}) - (1-\rho)(\bx\trans \bz) (\bz \trans \EE{\by \mid \bx, \bz})}\\
& = \frac{1}{\det}\p{ (1+\rho)(\bz\trans \bz)( \bx \trans \p{\thetas_x \bx + \thetas_z \bz}) - (1-\rho)(\bx\trans \bz) (\bz \trans \p{\thetas_x \bx + \thetas_z \bz})}\\
& = \frac{1}{\det}\p{ (1+\rho)\sqb{\thetas_x(\bx \trans \bx)(\bz\trans \bz) + \thetas_z(\bz\trans \bz)(\bx \trans \bz)} - (1-\rho)\sqb{\thetas_x(\bx\trans \bz)^2 + \thetas_z(\bx \trans \bz)(\bz\trans \bz)}}\\
& = \frac{1}{\det}\p{ 
\thetas_x\sqb{(\bx \trans \bx)(\bz\trans \bz) - (\bx\trans \bz)^2}
+ \rho \p{\thetas_x\sqb{(\bz\trans \bz)(\bx \trans \bx) + (\bx\trans \bz)^2}  + 2\thetas_z(\bx \trans \bz)(\bz\trans \bz) }}.
\end{split}
\end{equation}
Note that 
\begin{equation}
\begin{split}
\det 
&= (1+\rho)^2 (\bx\trans \bx) (\bz\trans \bz) - (1-\rho)^2 (\bx\trans \bz)^2\\
&= (\bx\trans \bx) (\bz\trans \bz) - (\bx\trans \bz)^2
+ 2\rho \sqb{(\bx\trans \bx) (\bz\trans \bz) + (\bx\trans \bz)^2}
+ \rho^2 \sqb{(\bx\trans \bx) (\bz\trans \bz) - (\bx\trans \bz)^2}.
\end{split}
\end{equation}
Therefore,
\begin{equation}
\begin{split}
 & \qquad   \EE{\thetah_x  - \thetas_x\mid \bx, \bz}\\
& = \frac{
\thetas_x\sqb{(\bx \trans \bx)(\bz\trans \bz) - (\bx\trans \bz)^2}
+ \rho \p{\thetas_x\sqb{(\bz\trans \bz)(\bx \trans \bx) + (\bx\trans \bz)^2}  + 2\thetas_z(\bx \trans \bz)(\bz\trans \bz) }}{(\bx\trans \bx) (\bz\trans \bz) - (\bx\trans \bz)^2
+ 2\rho \sqb{(\bx\trans \bx) (\bz\trans \bz) + (\bx\trans \bz)^2}
+ \rho^2 \sqb{(\bx\trans \bx) (\bz\trans \bz) - (\bx\trans \bz)^2}} - \thetas_x\\
& = \frac{
\rho \p{\thetas_x\sqb{-(\bz\trans \bz)(\bx \trans \bx) - (\bx\trans \bz)^2}  + 2\thetas_z(\bx \trans \bz)(\bz\trans \bz) } - \rho^2 \thetas_x\sqb{(\bx\trans \bx) (\bz\trans \bz) - (\bx\trans \bz)^2}
}{(\bx\trans \bx) (\bz\trans \bz) - (\bx\trans \bz)^2
+ 2\rho \sqb{(\bx\trans \bx) (\bz\trans \bz) + (\bx\trans \bz)^2}
+ \rho^2 \sqb{(\bx\trans \bx) (\bz\trans \bz) - (\bx\trans \bz)^2}}\\
& = \frac{b_1 \rho + a_1 \rho^2}{c_2 + b_2 \rho + a_2 \rho^2},
\end{split}
\end{equation}
where $b_1, a_1, c_2, b_2, a_2$ are expressions depending on $\bx$ and $\bz$ but not on $\rho$. We can then clearly see that when $\rho = 0$, $\EE{\thetah_x  - \thetas_x\mid \bx, \bz} = 0$. By symmetry, we have the same property for $\EE{\thetah_z - \thetas_z\mid \bx, \bz}$. Therefore,
\begin{equation}
\begin{split}
 \frac{d}{d\rho}B^2(\bx, \bz; \rho) &=
(\gx^2 + \sigma_x^2) \frac{d}{d\rho}\EE{\hat{\theta}_x - \thetas_x  \mid \bx, \bz}^2|_{\rho = 0} 
+  \p{\gz^2 + \sigma_z^2} \frac{d}{d\rho}\EE{\hat{\theta}_z - \thetas_z \mid \bx, \bz}^2|_{\rho = 0}\\
&\qquad \qquad + 2\p{\gx \gz} \frac{d}{d\rho}\p{\EE{\hat{\theta}_z - \thetas_z\mid \bx, \bz} \EE{\hat{\theta}_x - \thetas_x   \mid \bx, \bz}} |_{\rho = 0} \\
&= 2(\gx^2 + \sigma_x^2) \p{\EE{\hat{\theta}_x - \thetas_x  \mid \bx, \bz}|_{\rho = 0}} \frac{d}{d\rho}\EE{\hat{\theta}_x - \thetas_x  \mid \bx, \bz}|_{\rho = 0} \\
& \qquad \qquad +  2\p{\gz^2 + \sigma_z^2} \p{\EE{\hat{\theta}_z - \thetas_z \mid \bx, \bz}|_{\rho = 0}}\frac{d}{d\rho}\EE{\hat{\theta}_z - \thetas_z \mid \bx, \bz}|_{\rho = 0}\\
&\qquad \qquad + 2\p{\gx \gz} \p{\EE{\hat{\theta}_z - \thetas_z\mid \bx, \bz} |_{\rho = 0}}\frac{d}{d\rho} \EE{\hat{\theta}_x - \thetas_x   \mid \bx, \bz}|_{\rho = 0} \\
&\qquad \qquad + 2\p{\gx \gz} \p{\EE{\hat{\theta}_x - \thetas_x\mid \bx, \bz} |_{\rho = 0}}\frac{d}{d\rho} \EE{\hat{\theta}_z - \thetas_z   \mid \bx, \bz} |_{\rho = 0}.
\end{split}
\end{equation}
Since $\EE{\hat{\theta}_x - \thetas_x\mid \bx, \bz} |_{\rho = 0} = \EE{\hat{\theta}_z - \thetas_z\mid \bx, \bz} |_{\rho = 0} = 0$, we have 
\begin{equation}
    \frac{d}{d\rho}B^2(\bx, \bz; \rho) = 0. 
\end{equation}

It remains to study $ \frac{d}{d\rho}V(\bx, \bz; \rho)$. From the form of $\thetah_x$, we can get that
\begin{equation}
\begin{split}
\Var{\thetah_x \mid \bx, \bz} 
&= \frac{\sigma^{\star 2}}{\det^2} \norm{ (1+\rho)(\bz\trans \bz)\bx - (1-\rho)(\bx\trans \bz) \bz}_2^2\\
&= \frac{\sigma^{\star 2}}{\det^2} \sqb{(1+\rho)^2 (\bz\trans \bz)^2 (\bx\trans \bx) + (1- \rho)^2 (\bx\trans \bz)^2 (\bz\trans \bz) - 2(1+\rho)(1 - \rho) (\bx\trans \bz)^2 (\bz\trans \bz)}\\
& = \frac{\sigma^{\star 2} (\bz\trans \bz)}{\det^2} \p{ \sqb{(\bx\trans \bx)(\bz\trans \bz) - (\bx\trans \bz)^2}
+ 2 \sqb{(\bx\trans \bx)(\bz\trans \bz) - (\bx\trans \bz)^2}\rho
+ a_{v1} \rho^2},
\end{split}
\end{equation}
for some $a_{v1}$ depending on $\bx$ and $\bz$ but not on $\rho$. Similarly, we get that 
\begin{equation}
\begin{split}
\Var{\thetah_z \mid \bx, \bz} 
& = \frac{\sigma^{\star 2} (\bx\trans \bx)}{\det^2} \p{ \sqb{(\bx\trans \bx)(\bz\trans \bz) - (\bx\trans \bz)^2}
+ 2 \sqb{(\bx\trans \bx)(\bz\trans \bz) - (\bx\trans \bz)^2}\rho
+ a_{v2} \rho^2},
\end{split}
\end{equation}
for some $a_{v2}$ depending on $\bx$ and $\bz$ but not on $\rho$. For the covariance term, 
\begin{equation}
\begin{split}
\Cov{\thetah_x, \thetah_z \mid \bx, \bz} 
&= \frac{\sigma^{\star 2}}{\det^2} \sqb{ (1+\rho)(\bz\trans \bz)\bx - (1-\rho)(\bx\trans \bz) \bz}\trans \sqb{ (1+\rho)(\bx\trans \bx)\bz - (1-\rho)(\bz\trans \bx) \bx}\\
& = \frac{\sigma^{\star 2}}{\det^2}\p{(-1 + 3\rho)(1 + \rho)(\bx\trans \bx)(\bz\trans \bz)(\bx\trans \bz) + (1 - \rho)^2(\bx\trans \bz)^3}\\
& = \frac{\sigma^{\star 2} (\bx\trans \bz)}{\det^2} \p{ -\sqb{(\bx\trans \bx)(\bz\trans \bz) - (\bx\trans \bz)^2}
+ 2 \sqb{(\bx\trans \bx)(\bz\trans \bz) - (\bx\trans \bz)^2}\rho
+ a_{v3} \rho^2},
\end{split}
\end{equation}
for some $a_{v3}$ depending on $\bx$ and $\bz$ but not on $\rho$. 
Combining the three terms, we get
\begin{equation}
\label{eqn:variance_form}
\begin{split}
V(\bx, \bz; \rho) 
&= \Var{\hat{\theta}_x \mid \bx, \bz} (\gx^2 + \sigma_x^2) + \Var{\hat{\theta}_z \mid \bx, \bz} \p{\gz^2 + \sigma_z^2} + \Cov{\hat{\theta}_z, \hat{\theta}_x \mid \bx, \bz}  \p{\gx \gz}\\
& = \sigma^{\star 2} \frac{C_1 + B_1 \rho + A_1 \rho^2}{\det^2}
= \sigma^{\star 2} \frac{C_1 + B_1 \rho + A_1 \rho^2}{(C_2 + B_2 \rho + A_2 \rho^2)^2},
\end{split}
\end{equation}
where 
\begin{equation}
\begin{split}
C_1 &= \sqb{(\gx^2 + \sigma_x^2)(\bz \trans \bz) + (\gz^2 + \sigma_z^2)(\bx \trans \bx) - 2\gx \gz (\bx \trans \bz)}((\bx\trans \bx)(\bz\trans \bz) - (\bx\trans \bz)^2),\\
B_1 & = 2\sqb{(\gx^2 + \sigma_x^2)(\bz \trans \bz) + (\gz^2 + \sigma_z^2)(\bx \trans \bx) + 2\gx \gz (\bx \trans \bz)}((\bx\trans \bx)(\bz\trans \bz) - (\bx\trans \bz)^2),\\
C_2 &= (\bx\trans \bx)(\bz\trans \bz) - (\bx\trans \bz)^2,\\
B_2 &= 2\p{(\bx\trans \bx)(\bz\trans \bz) + (\bx\trans \bz)^2}. 
\end{split}
\end{equation}
By Lemma \ref{lemm:var_der}, 
$\frac{d}{d\rho}V(\bx, \bz; \rho)|_{\rho = 0} = \sigma^{\star 2} (C_2 B_1 - 2 C_1 B_2)/C_2^3$. 

Finally by \eqref{eqn:mse_decom}
\begin{equation}
\begin{split}
   \frac{d}{d\rho}\operatorname{MSE}(\bx, \bz; \rho) |_{\rho = 0} &= \frac{d}{d\rho}B^2(\bx, \bz; \rho)|_{\rho = 0} + \frac{d}{d\rho}V(\bx, \bz; \rho)|_{\rho = 0}  \\
&= \frac{d}{d\rho}V(\bx, \bz; \rho)|_{\rho = 0} = \sigma^{\star 2} (C_2 B_1 - 2 C_1 B_2)/C_2^3.
\end{split}
\end{equation}

\subsection{Proof of Proposition \ref{prop:dev_big_O}}
By the central limit theorem, we have that
\begin{equation}
\label{eqn:asym_C1_etc}
    \bx \trans \bx = n\p{\gx^2 + \sigma_x^2} + \oo_p\p{\sqrt{n}},
    \qquad 
    \bz \trans \bz = n\p{\gz^2 + \sigma_z^2} + \oo_p\p{\sqrt{n}},
    \qquad 
    \bx \trans \bz = n\p{\gx \gz} + \oo_p\p{\sqrt{n}}.
\end{equation}
Plugging into \eqref{eqn:C1_B2} gives
\begin{equation}
\begin{split}
C_1 &= 2 n^3 \sqb{(\gx^2 + \sigma_x^2)(\gz^2 + \sigma_z^2) - \gx^2 \gz^2 }^2 + \oo_p\p{n^{5/2}},\\
B_1 & = 4 n^3 \sqb{(\gx^2 + \sigma_x^2)(\gz^2 + \sigma_z^2) + \gx^2 \gz^2 }\sqb{(\gx^2 + \sigma_x^2)(\gz^2 + \sigma_z^2) - \gx^2 \gz^2 } + \oo_p\p{n^{5/2}},\\
C_2 &= n^2 \sqb{(\gx^2 + \sigma_x^2)(\gz^2 + \sigma_z^2) - \gx^2 \gz^2 } + \oo_p\p{n^{3/2}},\\
B_2 &= 2 n^2 \sqb{(\gx^2 + \sigma_x^2)(\gz^2 + \sigma_z^2) + \gx^2 \gz^2 } + \oo_p\p{n^{3/2}}. 
\end{split}
\end{equation}
Thus we have that
\begin{equation}
\begin{split}
C_1 B_2 &= 4 n^5 \sqb{(\gx^2 + \sigma_x^2)(\gz^2 + \sigma_z^2) + \gx^2 \gz^2 } \sqb{(\gx^2 + \sigma_x^2)(\gz^2 + \sigma_z^2) - \gx^2 \gz^2 }^2 + \oo_p\p{n^{9/2}},\\
C_2 B_1 & = 4 n^5 \sqb{(\gx^2 + \sigma_x^2)(\gz^2 + \sigma_z^2) + \gx^2 \gz^2 } \sqb{(\gx^2 + \sigma_x^2)(\gz^2 + \sigma_z^2) - \gx^2 \gz^2 }^2 + \oo_p\p{n^{9/2}}.\\
\end{split}
\end{equation}
Therefore, 
\begin{equation}
C_2 B_1 - 2C_1 B_2 = -4 n^5 \sqb{(\gx^2 + \sigma_x^2)(\gz^2 + \sigma_z^2) + \gx^2 \gz^2 } \sqb{(\gx^2 + \sigma_x^2)(\gz^2 + \sigma_z^2) - \gx^2 \gz^2 }^2 + \oo_p\p{n^{9/2}}. 
\end{equation}
Now we also know that 
\begin{equation}
C_2^3 = n^6 \sqb{(\gx^2 + \sigma_x^2)(\gz^2 + \sigma_z^2) - \gx^2 \gz^2 }^3 + \oo_p\p{n^{11/2}}.
\end{equation}
Hence
\begin{equation}
\begin{split}
&\frac{d}{d\rho}\sqb{\operatorname{MSE}(\bx, \bz; \rho)} |_{\rho = 0} \\
&\qquad = \sigma^{\star 2}(C_2 B_1 - 2 C_1 B_2)/C_2^3\\
& \qquad= \frac{-4 n^5 \sqb{(\gx^2 + \sigma_x^2)(\gz^2 + \sigma_z^2) + \gx^2 \gz^2 } \sqb{(\gx^2 + \sigma_x^2)(\gz^2 + \sigma_z^2) - \gx^2 \gz^2 }^2 }{n^6 \sqb{(\gx^2 + \sigma_x^2)(\gz^2 + \sigma_z^2) - \gx^2 \gz^2 }^3}\sigma^{\star 2} + \oo_p\p{n^{-3/2}}\\
& \qquad= -\frac{4}{n} \frac{(\gx^2 + \sigma_x^2)(\gz^2 + \sigma_z^2) + \gx^2 \gz^2}{(\gx^2 + \sigma_x^2)(\gz^2 + \sigma_z^2) - \gx^2 \gz^2}\sigma^{\star 2} + \oo_p\p{n^{-3/2}}\\
& \qquad= -\frac{4}{n}\p{1 + \frac{2 \gx^2\gz^2}{\sigma_x^2\gz^2 + \sigma_z^2\gx^2 + \sigma_x^2 \sigma_z^2}}\sigma^{\star 2} + \oo_p\p{n^{-3/2}}\\
& \qquad=  -\frac{4}{n}\p{1 + \frac{2 \gx^2 \gz^2}{ \sigma_x^2\gz^2 + \sigma_z^2 \gx^2 + \sigma_x^2 \sigma_z^2 }}\p{\sigma_y^2 + \frac{\gy^2 \sigma_x^2 \sigma_z^2}{\sigma_x^2\gz^2 + \sigma_z^2 \gx^2 + \sigma_x^2 \sigma_z^2} } + \mathcal{O}_p\p{n^{-3/2}}. 
\end{split}
\end{equation}

\subsection{Proof of Proposition \ref{prop:dev_ex_MSE}}
For $\operatorname{MSE}(\bx, \bz; 0)$, we have that by \eqref{eqn:mse_decom}, 
\begin{equation}
    \operatorname{MSE}(\bx, \bz; 0) 
= B^2(\bx, \bz; 0) + V(\bx, \bz; 0) + \sigma^{\star 2}
= V(\bx, \bz; 0) + \sigma^{\star 2}. 
\end{equation}
Here we make use of the fact that when $\rho = 0$, $\EE{\thetah_x  - \thetas_x\mid \bx, \bz} = \EE{\thetah_z - \thetas_z\mid \bx, \bz} = 0$ and that $B^2(\bx, \bz; 0) = 0$. For $V(\bx, \bz; 0)$, we have that by \eqref{eqn:variance_form} and \eqref{eqn:asym_C1_etc},
\begin{equation}
V(\bx, \bz; 0) 
= \sigma^{\star 2} \frac{C_1 }{C_2^2}
= \frac{4 \sigma^{\star 2}}{n}  + \oo_p\p{n^{-3/2}}
= \oo_p\p{\frac{1}{n}}. 
\end{equation}
Therefore,
\begin{equation}
\operatorname{MSE}(\bx, \bz; 0) = V(\bx, \bz; 0) + \sigma^{\star 2} = \sigma^{\star 2} +  \oo_p\p{\frac{1}{n}}. 
\end{equation}
Thus, together with the result in Proposition \ref{prop:dev_big_O}, we have 
\begin{equation}
\frac{\frac{d}{d\rho}\sqb{\operatorname{MSE}(\bx, \bz; \rho)} |_{\rho = 0}}{\operatorname{MSE}(\bx, \bz; 0)} 
= - \frac{4}{n}\p{1 + \frac{2 \gx^2 \gz^2}{ \sigma_x^2\gz^2 + \sigma_z^2 \gx^2 + \sigma_x^2 \sigma_z^2 }} + \mathcal{O}_p\p{n^{-\frac{3}{2}}}. 
\end{equation}

\section{Distribution of predicted versus true time to delivery for the labor onset prediction example}

We show in Figure \ref{fig:distribution} the distribution of predicted and true time to delivery for each patient, which gives a better sense of the quality of the predictions for the regression task.
The left plot shows the distribution of time to delivery for all patients at their first time points in the longitudinal study; the right plot shows the predicted versus true time to delivery for the training and test samples. This is based on one random split of the training and test sets of 40 and 13 patients, respectively.

\begin{figure}[h] 
\includegraphics[width=0.9\textwidth]{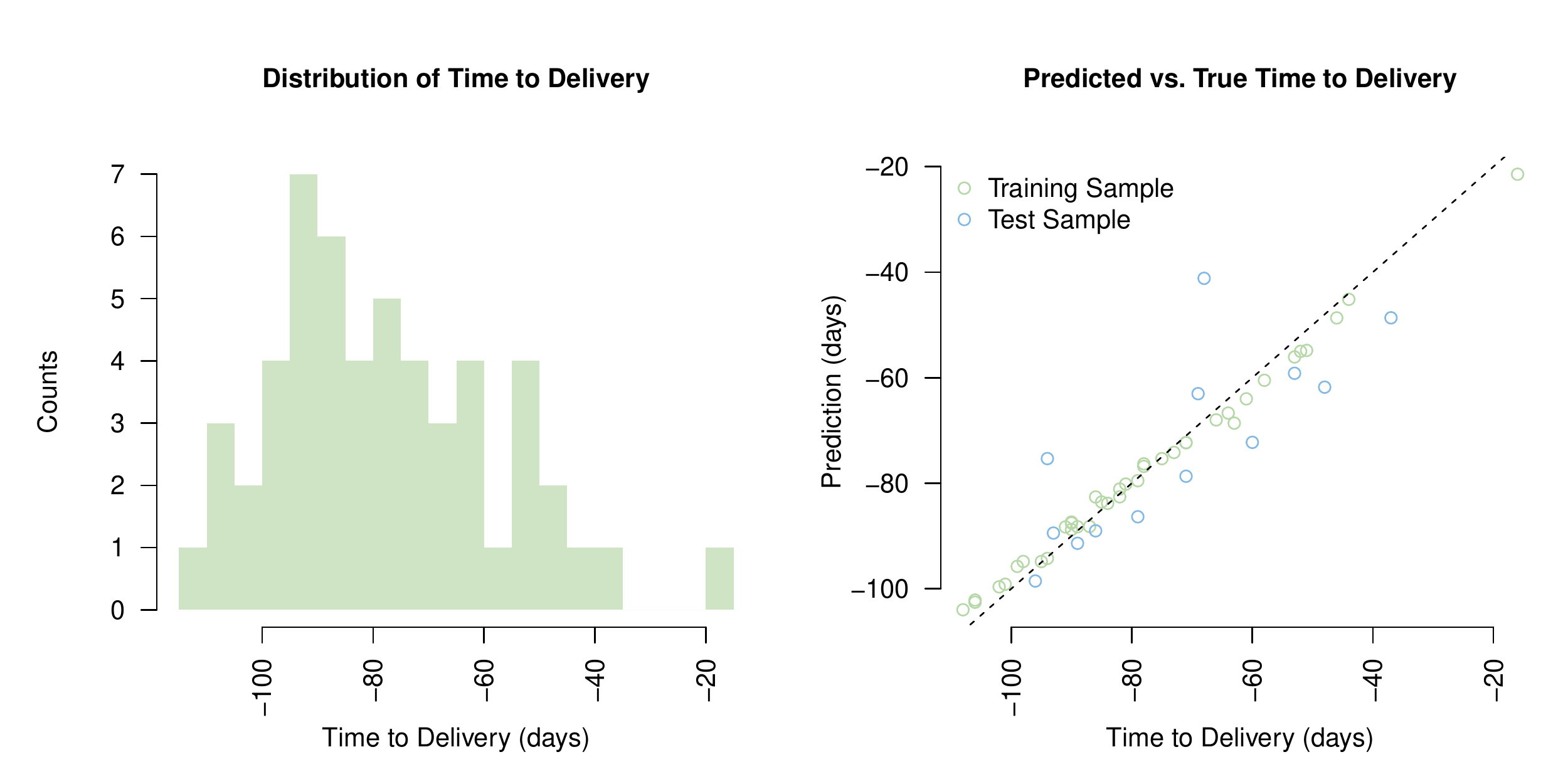}
    \caption{{\em Distribution of time to delivery and predicted versus true time to delivery for training and test samples.} The predictions were derived from cooperative learning.}
    \label{fig:distribution}
\end{figure}

\section{Procedure for generating the imaging and ``omics'' data}
\label{sec:appendix_imaging}

Here we outline the detailed procedure for data generation in the simulation study with imaging and ``omics'' data in Algorithm \ref{alg:simulation_two_modality}.
The ``omics'' data ($X$), imaging data ($Z$), and the response $\y$ are generated such that there are correlations between $X$, $Z$, and $\y$.

\vspace{6mm}

\begin{algorithm}[H]
\KwIn{Parameters $n, p_{x}, p_{u}, s_{u}, t, \sigma, \bbeta_{u}, I_{\text{max}}, \text{ndim}, \text{threshold}$.}
\KwOut{$X \in \mR^{n\times p_{x}}$ (omics), $Z \in \mR^{n \times \text{ndim}\times \text{ndim} \times1}$ (images assuming one color channel), $\y \in \mR^{n}$.}
\begin{enumerate}
    \item $x_j \in \mR^{n}$ distributed i.i.d. MVN$(0, I_{n})$ for $j = 1,2,\ldots,p_{x}$
    \item For $i = 1, 2, \ldots, p_{u}$ ($p_{u} < p_{x}$, where $p_{u}$ corresponds to the number of factors):
    \begin{enumerate}
        \item $u_{i} \in \mR^{n}$ distributed i.i.d. MVN$(0, s_{u}^2I_{n})$ 
        \item $x_{i} = x_{i} + t * u_{i}$
    \end{enumerate}
    \item  $U = [u_{1}, u_{2}, \ldots, u_{p_{u}}]$, $X = [x_{1}, x_{2}, \ldots, x_{p_{x}}]$
    \item $\y_{u} = U\bbeta_{u} + \epsilon$ where $\epsilon \in \mR^{n}$ distributed i.i.d. MVN$(0, \sigma^2 I_{n})$
    \item For $i = 1, 2, \ldots, n$:
    \begin{enumerate}
        \item $P_{i} = \frac{1}{1+\exp(\y_{u_{i}})}$, $\y_{i} \sim \text{Bernoulli}(P_{i})$
        \item Generate a 2D pixel matrix of image $Z_{i} \in \mR^{\text{dim} \times \text{dim} \times 1}$
        \item Generate a polygon $\text{PG}_{i}$ inside $Z_{i}$, defined by 4 vertices [$v_{1}, v_{2}, v_{3}, v_{4}$] on the axes, i.e. $v_{1} = [0,a_{1}], v_{2} = [0,a_{2}], v_{3} = [a_{3},0], v_{4} = [a_{4},0]$, where $a_{1} \sim \text{Unif}(\frac{\text{ndim}}{2}, \text{ndim}), a_{2} \sim \text{Unif}(-\text{ndim}, -\frac{\text{ndim}}{2}), a_{3} \sim \text{Unif}(\frac{\text{ndim}}{2}, \text{ndim}), a_{4} \sim \text{Unif}(-\text{ndim}, -\frac{\text{ndim}}{2})$
        \item Randomly sample points from $Z_{i}$: if the point $[x', y']$ falls inside the polygon $\text{PG}_{i}$, i.e.  $[x', y'] \in \text{PG}_{i}$, then $Z_{i}[x', y'] \sim \text{Unif}(0,1)$
        \item If $\y_{i} = 1$, $I_{\text{disease}} = I_{\text{max}} \times \y_{u_{i}}$, where $I_{\text{max}}$ is the maximum intensity of pixel values for images,
        \begin{itemize}
            \item For $x' = 1,2,\ldots,\text{ndim}$:
            \begin{itemize}
                \item For $y' = 1,2,\ldots,\text{ndim}$:
                \begin{itemize}
                    \item $P(x', y') \sim \text{Unif} (0,1)$
                    \item If $[x', y'] \in \text{PG}_{i}$ and $P(x', y') <$ threshold, $Z_{i}[x', y'] = I_{\text{disease}}$
                \end{itemize}
            \end{itemize}
        \end{itemize}
    \end{enumerate}
\end{enumerate}
\caption{\em Simulation procedure for generating the imaging and ``omics'' data.} \label{alg:simulation_two_modality}
\end{algorithm}

\end{appendix}
\end{document}